\documentclass[journal]{IEEEtran}

\usepackage{graphicx}
\usepackage{cite}
\usepackage{amsmath}
\usepackage{amsthm}
\usepackage{amssymb}
\usepackage{physics}
\usepackage{mathtools}
\usepackage{mleftright}

\newtheorem{theorem}{Theorem}
\newtheorem{corollary}{Corollary}[theorem]
\newtheorem{proposition}[theorem]{Proposition}

\newtheorem*{theorem*}{Theorem}
\newtheorem*{corollary*}{Corollary}
\newtheorem*{proposition*}{Proposition}

\theoremstyle{definition}

\newtheorem{definition}{Definition}
\newtheorem*{remark*}{Remark}

\DeclareMathOperator{\modulo}{mod}
\DeclareMathOperator*{\argmin}{argmin}

\DeclareMathOperator{\Expect}{E}
\DeclarePairedDelimiter{\ceil}{\lceil}{\rceil}
\DeclarePairedDelimiter{\floor}{\lfloor}{\rfloor}
\def\mmiddle#1{\mathrel{}\middle#1\mathrel{}}

\allowdisplaybreaks

\begin{document}
\title{Space--Time Codes from Sum-Rank Codes}
\author{Mohannad~Shehadeh,~\IEEEmembership{Graduate Student Member,~IEEE}
        and~Frank~R.~Kschischang,~\IEEEmembership{Fellow,~IEEE}
\thanks{The authors are with the Edward S. Rogers Sr. Department of Electrical and Computer Engineering, University of Toronto, Toronto, ON M5S 3G4, Canada (emails: \{mshehadeh, frank\}@ece.utoronto.ca).
}
\thanks{Parts of this work were presented at the 2020 IEEE International Symposium on Information Theory (ISIT) \cite{ME}.}
}
\maketitle

\begin{abstract}
Just as rank-metric or Gabidulin codes may be used to construct
rate--diversity tradeoff optimal space--time codes, a recently introduced
generalization for the sum-rank metric---linearized Reed--Solomon
codes---accomplishes the same in the case of multiple fading blocks. In
this paper, we provide the first explicit construction of minimal delay
rate--diversity optimal multiblock space--time codes as an application of
linearized Reed--Solomon codes.  We also provide sequential decoders for
these codes and, more generally, space--time codes constructed from finite
field codes.  Simulation results show that the proposed codes can
outperform full diversity codes based on cyclic division algebras at low
SNRs as well as utilize significantly smaller constellations.
\end{abstract}

\begin{IEEEkeywords}
Rank-metric codes, space--time codes, sum-rank codes, wireless communication.
\end{IEEEkeywords}

\IEEEpeerreviewmaketitle

\section{Introduction}

\IEEEPARstart{T}{his} paper builds upon a line of work which considers the
design of space--time codes that optimally trade off diversity for rate at
a fixed constellation size.  Our primary contributions are as follows:

1) By replacing Gabidulin codes \cite{Gabidulin-Codes} in known
rate--diversity optimal space--time code constructions
\cite{ST-Gaussian-Integers, Gabidulin-Space-Time,Sven,Arbitrary,Lu-Kumar}
with \textit{linearized Reed--Solomon codes} \cite{LRS-Codes}, we obtain
the first explicit construction of minimal delay rate--diversity optimal
\textit{multiblock} space--time codes. This provides the first general
solution to a problem first posed in \cite{Rate-Diversity-Tradeoff-General}
and \cite{Lu-Kumar}.

2) We provide sequential maximum likelihood (ML) decoders for these codes.
More generally, we show that many sequential decoding strategies for
space--time codes
\cite{Sphere-Decoder,Sphere-Decoder-2,ML-CLPS,Tree-Search-Decoding,Spherical-Bound-Stack-Decoder,SDP-Inspired-Lower-Bounds,
Eigenbound-Equivalent} that are typically thought to be only applicable to
codes with a linear dispersion form \cite{Linear-Dispersion} can in some
cases be effectively adapted for use with the proposed codes and similarly
constructed codes \cite{ST-Gaussian-Integers,
Gabidulin-Space-Time,Sven,Arbitrary}.

3) Facilitated by these ML decoders, we provide an empirical study of the
performance of the proposed codes in simulation as well as related codes
\cite{ST-Gaussian-Integers,Gabidulin-Space-Time,Sven,Arbitrary} which had
not been previously decoded for large codebook sizes.  This demonstrates
that these codes can outperform full diversity codes
\cite{Hsiao-Code,Sheng,Perfect-Codes,Improved4} based on cyclic division
algebras (CDAs) \cite{Division-Algebras} at low SNRs and using smaller
constellations.

We emphasize that the latter two contributions cover new ground in the
single-block setting as well.  Apart from the primary contributions, we
consolidate some results and observations occurring in some of the previous
literature on the rate--diversity optimal space--time coding problem
\cite{ST-Gaussian-Integers,Gabidulin-Space-Time,Sven,
Arbitrary,Lu-Kumar,Rate-Diversity-Tradeoff-General} and attempt to situate
these lines work within the broader literature on space--time coding.  We
particularly note that this paper contains the first error performance
comparison of codes designed for rate--diversity tradeoff optimality with
codes designed for diversity--multiplexing tradeoff optimality. 

The remainder of this paper is organized as follows: Section \ref{Ch2-Sec}
establishes the setting, introduces the rate--diversity optimal multiblock
space--time coding problem, and briefly surveys prior work on the problem.
Section \ref{Other-Perspectives-Sec} discusses the relevance of the
rate--diversity perspective and provides an alternative interpretation of
the rate--diversity tradeoff to aid in comparing with codes designed from
other perspectives.  Section \ref{Code-Construction-Sec} provides the
proposed code construction after introducing the required technical
ingredients which are \textit{rank-metric-preserving maps}
\cite{ST-Gaussian-Integers,Gabidulin-Space-Time,
Gauss-Int-Map-Is-RMP,Arbitrary,Sven,Lu-Kumar} and \textit{linearized
Reed--Solomon codes} \cite{LRS-Codes}.  The error performance of the
proposed codes is studied in simulation in Section \ref{Simulations-Sec}
with the subject of their decoding deferred to Section \ref{Decoding-Sec}.
Section \ref{Conclusion-Sec} concludes the paper with some suggestions for
future work.

\section{Setting, Problem Statement, and Basic Results}\label{Ch2-Sec}

\subsection{Channel Model}

Adopting the setting and conventions of \cite{Lu-Kumar},
we consider a multiple-input multiple-output 
(MIMO) Rayleigh block-fading channel with 
$n_t$ transmit antennas, $n_r$ receive antennas, 
and $L$ fading blocks per codeword each static
for duration $T$. 
An \textit{$L$-block $n_t \times T$ space--time code}
$\mathcal{X}$ is a finite subset of
$\mathbb{C}^{n_t \times LT}$
of cardinality greater than or equal to two.
A \textit{codeword} $X \in \mathcal{X}$
is a complex 
$n_t \times LT$ matrix 
${X = \begin{bmatrix} X_1 & X_2 & \cdots & X_L \end{bmatrix}}$
which partitions into $L$ sub-matrices $X_1,X_2,\dots,X_L$
of dimensions ${n_t\times T}$ referred
to as \textit{sub-codewords.}
For ${\ell=1,2,\dots,L}$,
$Y_\ell$ is the $n_r \times T$
received matrix given by
\begin{equation}\label{Channel}
Y_\ell = \rho H_\ell X_\ell + W_\ell
\end{equation}
where $H_\ell$ is the ${n_r \times n_t}$ channel matrix and
$W_\ell$ is the ${n_r \times T}$ noise matrix
with both having iid 
circularly-symmetric complex Gaussian entries
with unit variance. The codeword $X$
is sampled uniformly at random
from a code $\mathcal{X}$ and 
the real scalar parameter 
$\rho$ is chosen
to satisfy 
\begin{equation}
\Expect\mleft[\lVert \rho X \rVert_\mathsf{F}^2\mright] 
= 
\rho^2 \sum_{\ell=1}^L
\Expect\mleft[ 
\lVert X_\ell \rVert_\mathsf{F}^2\mright]
= 
L \cdot T \cdot \mathsf{SNR}\text{.}
\end{equation}

We further have ML decoding at the receiver
with perfect channel knowledge in the sense
that all channel matrix realizations are known
and all channel model parameters are known.
In particular, define the ML decision $\hat{X}$
by
\begin{equation}
\hat{X} = \argmin_{X' \in \mathcal{X}}
\sum_{\ell=1}^{L}\lVert Y_\ell - \rho H_\ell X'_\ell
\rVert_\mathsf{F}^2\label{ML-Estimate}
\end{equation}
and define the probability of error $P_e$
by 
\begin{equation}
	P_e = \Pr(\hat{X}\neq X)\text{.}
\end{equation}
The decoding problem to 
be considered in Section \ref{Decoding-Sec} 
is that of 
solving $\eqref{ML-Estimate}$
and each simulation curve
in Section
\ref{Simulations-Sec}
provides the codeword
error rate (CER)
which is a Monte Carlo estimate of $P_e$
as a function of SNR.

Throughout this paper,
we will use $\intercal$ for the
matrix transpose, $\dagger$ for the
Hermitian transpose, and $(\cdot)_{ij}$
to denote the $ij$th entry of a matrix.

\begin{definition}[Transmit Diversity Gain]
	\label{Transmit-Diversity-Gain-Definition}
	An $L$-block $n_t \times T$ space--time code is
	said to achieve a
	\textit{transmit diversity gain} of $d$ if,
	for a
	channel with $n_r$ receive antennas, we have 
	\begin{equation}
	\lim_{\mathsf{SNR}\to\infty}
	\frac{
		\log P_e(\mathsf{SNR})
	}
	{\log\mathsf{SNR}}
	= -n_r d\text{.}
	\end{equation}
\end{definition}

\begin{theorem}[Sum-Rank Criterion
	\cite{Transmit-Diversity,Special-Case-Conf,Lu-Kumar,
		Rate-Diversity-Tradeoff-General,ST,Hsiao-Lower-Bound,
		Hsiao-Lower-Bound-Conf}]
	\label{Sum-Rank-Criterion-Theorem}
	An $L$-block $n_t \times T$ space--time code $\mathcal{X}$ 
	achieves 
	a transmit diversity gain of $d$ if and only if 
	\begin{equation}\label{Complex-Sum-Rank-Distance}
	d = \min_{\substack{X,X' \in \mathcal{X}\\X\neq X'}}
	\sum_{\ell=1}^{L}\rank(X_\ell-X'_\ell)\text{.}
	\end{equation}
\end{theorem}

\begin{proof}
	This follows from combining the pairwise error
	probability upper bounds of \cite{Transmit-Diversity,ST}
	with the lower bounds of  \cite{Hsiao-Lower-Bound,Hsiao-Lower-Bound-Conf}
	and then sandwiching the probability of error
	via the union bound. For a detailed proof, see \cite{My-Thesis}. 
\end{proof}

Theorem \ref{Sum-Rank-Criterion-Theorem} thus
provides an equivalent definition
of transmit diversity gain
as a property of the code used. 
The recognition of the quantity
\eqref{Complex-Sum-Rank-Distance} 
as a criterion for space--time code design
first occurs in \cite{Transmit-Diversity,ST}
and the multiblock
generalization which we adopt 
was first recognized in 
\cite{Special-Case-Conf, Lu-Kumar, Rate-Diversity-Tradeoff-General}.
The equivalence to Definition \ref{Transmit-Diversity-Gain-Definition}
follows from \cite{Hsiao-Lower-Bound,Hsiao-Lower-Bound-Conf}
which provide the required lower bound
and state this equivalence.

From \eqref{Complex-Sum-Rank-Distance}, we
have that the transmit diversity gain 
$d$ is an integer satisfying 
${1 \leq d\leq L\cdot\min\{n_t,T\}}$. 
Codes for which $d = L\cdot\min\{n_t,T\}$
are referred to as being
\textit{full diversity.} 
Moreover, codes for which $T = n_t$
 are referred to as being 
\textit{minimal delay.}

A \textit{constellation}
$\mathcal{A}$ is defined as a finite
subset of $\mathbb{C}$ of cardinality 
greater than or equal to two. 
Let $\mathcal{A}$ be a constellation. 
An $L$-block $n_t\times T$
space--time code $\mathcal{X}$ will be said to be
\textit{completely over}
$\mathcal{A}$
if $\mathcal{X}$ is
a subset of $\mathcal{A}^{n_t \times LT}$.
The term \textit{completely over} 
is adopted from \cite{Division-Algebras} 
and emphasizes that codeword entries
are constrained
to belong to $\mathcal{A}$.

\begin{definition}[Rate]
	For some 
	fixed constellation
	$\mathcal{A}$,
	the \textit{rate} $R$ of an $L$-block $n_t \times T$ space--time code 
	$\mathcal{X}$ completely over $\mathcal{A}$ is defined 
	by
	\begin{equation}\label{Rate-Definition}
		R = 
		\frac{1}{LT}\log_\abs*{\mathcal A}\abs*{\mathcal{X}}
		\text{.}
	\end{equation}
\end{definition}
The term 
\textit{channel use} will refer to a use of the 
underlying MIMO 
channel so that a codeword is
 transmitted across $LT$ channel uses and 
this rate is interpreted as the
average information rate in symbols 
per channel use.

\subsection{The Rate--Diversity Tradeoff and Some Consequences}

The following theorem first appears in
\cite{Rate-Diversity-Tradeoff-General,Lu-Kumar} 
and is a generalization of the 
well-known tradeoff for the case of $L = 1$ 
appearing in \cite{ST,Rate-Diversity-Tradeoff-Lu}.
It follows from a Singleton bound \cite{Singleton}
argument. 

For the remainder of this section,
we will let $\mathcal{A}$ 
be some fixed but arbitrary 
constellation.

\begin{theorem}[Rate--Diversity Tradeoff 
	\cite{Rate-Diversity-Tradeoff-General,Lu-Kumar}]
	\label{Rate-Diversity-Tradeoff-Theorem}
	Let $\mathcal{X}$ be an $L$-block $n_t \times T$  
	space--time code completely over $\mathcal{A}$ 
	with rate $R$ and achieving transmit diversity gain $d$. 
	Then, 
	\begin{equation}\label{Rate-Diversity-Tradeoff-Inequality}
	R \leq n_t - \frac{d-1}{L}\cdot
	\max\mleft\{\frac{n_t}{T},1\mright\}\text{.}
	\end{equation}
\end{theorem}

A space--time code
is said to be \textit{rate--diversity optimal}
if \eqref{Rate-Diversity-Tradeoff-Inequality}
holds with equality. 
A \textit{rate--diversity pair} $(R,d)$
for which \eqref{Rate-Diversity-Tradeoff-Inequality}
holds with equality
(in which case specifying one of $R$ or $d$ specifies the other) 
will be said to be an \textit{optimal rate--diversity pair.} 
The rate--diversity optimal multiblock
space--time coding problem is that 
of constructing families
of space--time codes capable of achieving 
\textit{any} optimal rate--diversity pair $(R,d)$
with $d$ a positive integer satisfying 
$1\leq d \leq L\cdot\min\{n_t,T\}$. 
We will be particularly interested
in codes that are \textit{not} 
full diversity. As will be seen
shortly, the extremes of this tradeoff
are relatively uninteresting. 

The remainder of this section will
consider what happens for some special
cases of the parameters $L$, $n_t$, $T$,
and $d$.  
The results which follow 
will not play a role
in the main code construction
to be provided in Section
\ref{Code-Construction-Sec}
which will admit arbitrary 
values for these parameters.
However, they are worth noting
both as basic consequences of Theorem
\ref{Rate-Diversity-Tradeoff-Theorem}
and for the purposes
of contextualizing the problem
at hand.

\begin{remark*}
	Let $\mathcal{X} = \mathcal{A}^{n_t \times LT}$ 
	be an $L$-block $n_t \times T$
	space--time code completely over $\mathcal{A}$. It is easy to see that this
	code, which corresponds to uncoded signalling, achieves the optimal rate--diversity pair $(R,d)$ corresponding to $d = 1$.
	Thus, restrictions to $d > 1$ are without elimination of interesting cases.
\end{remark*}

Solutions to the single-block, i.e., $L = 1$, rate--diversity 
optimal space--time coding problem  
are provided in \cite{ST-Gaussian-Integers,Gabidulin-Space-Time,
	Lu-Kumar,Sven,Lu-AM-PSK,Arbitrary,
	Hammons-Rate-Diversity-Optimal,Rate-Diversity-Tradeoff-Lu},
hence we shift our attention to the multiblock, i.e., $L > 1$, problem.
There are three special cases in which the multiblock problem is solved 
by a straightforward adaptation of a solution to the single-block problem. 
These cases are
\begin{itemize}
	\setlength\itemsep{1ex}
	\item $d = L\cdot\min\{n_t,T\}$, full diversity;
	\item $T \geq Ln_t$, wide sub-codewords or very slow fading; and
	\item $n_t \geq LT$, tall sub-codewords or very fast fading.
\end{itemize}
These three special
cases are hence implicitly solved. 
We show this in the next three propositions which deal with the construction
of optimal multiblock codes \textit{assuming that
	optimal single-block codes are at hand.}

In the case of $d= L\cdot \min\{n_t,T\}$, 
the only possibility admitted by \eqref{Complex-Sum-Rank-Distance}
is that all sub-codeword differences corresponding to distinct 
codewords are full rank. Thus,
as noted in \cite{Rate-Diversity-Tradeoff-General}
and as will be seen shortly,
repetition of a full diversity
single-block optimal code is optimal. On the other hand, in
the case of $d < L\cdot \min\{n_t,T\}$, \eqref{Complex-Sum-Rank-Distance}
admits a combinatorially vast space of possibilities for the 
ranks of the differences of the sub-codewords and repetition
of any single-block code only allows for the one
where all the ranks in the sum in \eqref{Complex-Sum-Rank-Distance}
are equal. Repetition in this case hence yields a lower
rate than is possible and is suboptimal. 
More precisely,
we have the the following proposition:

\begin{proposition}[Repetition Constructions]\label{Repetition-Construction}
	Let $\mathcal{\tilde{X}}$ be a
	rate--diversity optimal
	$1$-block $n_t \times T$ space--time code
	completely over $\mathcal{A}$ achieving 
	the
	(optimal) rate--diversity pair 
	$(\tilde R, \tilde d)$. 
	Let $\mathcal{X}$ be the $L$-block $n_t \times T$ space--time code 
	completely over $\mathcal{A}$
	with $L > 1$
	obtained by horizontally 
	concatenating 
	$L$ copies of each codeword of 
	$\mathcal{\tilde{X}}$  
	and achieving rate--diversity pair $(R,d)$.
	Then, $R = \tilde R / L$, $d = L\tilde d$,
	and $\mathcal{X}$ is rate--diversity
	optimal if and only if $\tilde d = \min\{n_t,T\}$.
\end{proposition}
\begin{proof}
	Since 
	$\abs{\mathcal{X}} = \abs*{\mathcal{\tilde X}}$,
	it is immediate from Definition \ref{Rate-Definition}
	that $R = \tilde R / L$.
	Moreover, noting that the
	sub-codewords of $\mathcal{X}$ are identical
	and are the codewords of $\mathcal{\tilde X}$, it
	is straightforwardly verified that $d = L\tilde d$.
	Next, note that by the rate--diversity
	optimality of $\mathcal{\tilde{X}}$, we have
	\begin{align*}
		\tilde R
		&= 
		n_t - (\tilde d - 1)
		\cdot
		\max\mleft\{\frac{n_t}{T},1\mright\}\\
		&= 
		n_t\cdot\left(1-\frac{\tilde d}{\min\{n_t,T\}}\right)	
		+
		\max\mleft\{\frac{n_t}{T},1\mright\}\text{.}
	\end{align*}
	We then have by
	\eqref{Rate-Diversity-Tradeoff-Inequality}
	\begin{align*}
		R = \frac{\tilde R}{L} &= 
		\frac{n_t}{L}\cdot\left(1-\frac{\tilde d}{\min\{n_t,T\}}\right)
		+
		\frac{1}{L}\cdot\max\mleft\{\frac{n_t}{T},1\mright\}\\
		&\leq 
		\frac{n_t}{1}\cdot\left(1-\frac{\tilde d}{\min\{n_t,T\}}\right)	
		+
		\frac{1}{L}\cdot\max\mleft\{\frac{n_t}{T},1\mright\}\\
		&=
		n_t - \frac{d-1}{L}\cdot \max\mleft\{\frac{n_t}{T},1\mright\}
	\end{align*}
	with equality if $\tilde d = \min\{n_t,T\}$.
	Conversely, equality implies that $\tilde d = \min\{n_t,T\}$
	by the $L > 1$ assumption.
\end{proof}

We now proceed to
the cases of $T \geq Ln_t$ and $n_t \geq LT$
which can be addressed with
what we term 
\textit{slicing constructions} to be described in the next two 
propositions. 
The first slicing construction to be described in 
Proposition \ref{Horizontal-Slicing-Construction} first occurs in
\cite{Lu-Kumar,Special-Case-Conf}. 
Proposition \ref{Vertical-Slicing-Construction} 
is a space--time coding analogue of a result in \cite{Umberto-MR-LRC}. 
Both are based on the simple fact that
\begin{equation*}
\rank\mleft(\begin{bmatrix}
A \\ B
\end{bmatrix}\mright)
\leq \rank(A) + \rank(B)
\end{equation*} 
for arbitrary matrices $A$ and $B$
over any field
having the same number of columns.
Similarly, 
we have 
\begin{equation}\label{Finer-Partitions-Preserve-Total-Rank}
\rank\mleft(\begin{bmatrix}
A & B
\end{bmatrix}\mright)
\leq \rank(A) + \rank(B)
\end{equation}
for matrices $A$ and $B$
with the same number of 
rows.

\begin{proposition}[Horizontal Slicing
	Construction \cite{Lu-Kumar,Special-Case-Conf}]
	\label{Horizontal-Slicing-Construction}
		Let $\mathcal{\tilde{X}}$ 
		be a rate--diversity 
		optimal $1$-block $Ln_t \times T$
		space--time code completely over $\mathcal{A}$
		with $L > 1$
		and let $\mathcal{X}$
		be the \mbox{$L$-block}
		$n_t \times T$ space--time code 
		completely over $\mathcal{A}$ 
		obtained by horizontally 
		slicing the codewords of 
		$\mathcal{\tilde{X}}$ into $L$ sub-codewords
		of dimensions $n_t \times T$. 
		If $T \geq Ln_t$, then $\mathcal{X}$
		is rate--diversity optimal.
\end{proposition}
\begin{proof}
	Let $(\tilde R, \tilde d)$ be the
	rate--diversity pair achieved by $\mathcal{\tilde X}$
	and $(R,d)$ be the rate--diversity pair achieved by 
	$\mathcal{X}$.
	By the rate--diversity optimality of $\mathcal{\tilde X}$,
	we have $\tilde R = Ln_t - \tilde d + 1$
	and since $\abs{\mathcal{X}}=\abs*{\mathcal{\tilde X}}$,
	we have 
	$R = \tilde R/L = n_t - (\tilde d - 1)/L$.
	Noting that $T \geq Ln_t \geq n_t$,
	it suffices to show that $d = \tilde d$.
	Rearranging the rate--diversity tradeoff
	\eqref{Rate-Diversity-Tradeoff-Inequality} 
	for $\mathcal{X}$, we have
	$d \leq Ln_t - LR + 1 = \tilde d$. Let $X,X'\in \mathcal{X}$ be a 
	codeword pair such that $d = \sum_{\ell=1}^L \rank(X_\ell-X'_\ell)$. Then, 
	\begin{align*}
	d &= \sum_{\ell=1}^L \rank(X_\ell-X'_\ell) \\
	& \geq \rank\mleft(\begin{bmatrix}
	X_1-X'_1\\
	X_2-X'_2\\
	\vdots\\
	X_L-X'_L
	\end{bmatrix}\mright)\\
	& \geq \tilde d
	\end{align*}
	since the vertically concatenated sub-codewords 
	of $\mathcal{X}$ are a codeword 
	of $\mathcal{\tilde{X}}$ by definition.
\end{proof}

We can similarly show the following: 

\begin{proposition}[Vertical Slicing
	Construction]
	\label{Vertical-Slicing-Construction}
	Let $\mathcal{\tilde{X}}$ 
	be a rate--diversity 
	optimal $1$-block $n_t \times LT$
	space--time code completely over $\mathcal{A}$
	with $L > 1$
	and let $\mathcal{X}$
	be the \mbox{$L$-block}
	$n_t \times T$ space--time code 
	completely over $\mathcal{A}$ 
	obtained by vertically
	slicing the codewords of 
	$\mathcal{\tilde{X}}$ into $L$ sub-codewords
	of dimensions $n_t \times T$. 
	If $n_t \geq LT$, then $\mathcal{X}$
	is rate--diversity optimal.
\end{proposition}

We conjecture that the converses of 
Propositions \ref{Horizontal-Slicing-Construction}
and \ref{Vertical-Slicing-Construction} are also
true. Weaker versions of the converse statements
can be found in \cite{My-Thesis}. 
Essentially, we cannot expect slicing to work
beyond the special cases of 
$T \geq Ln_t$ and $n_t \geq LT$
because
this technique inherently relies on the unnecessarily 
strong requirement of 
linear independence \textit{across} different sub-codeword matrices.
As a result, sub-codewords must be sufficiently wide or sufficiently
tall so that linear independence across them can be 
imposed without a rate penalty.

We seek a 
rate--diversity optimal family of 
multiblock space--time codes 
capable of achieving \textit{any} 
optimal rate--diversity pair $(R,d)$ 
with $d$ an integer satisfying ${1 \leq d \leq L\cdot\min\{n_t,T\}}$
for \textit{any} $L$, $T$, and $n_t$. 
To the best of the authors' knowledge, 
no such codes exist in the prior literature; 
existing constructions are either non-explicit 
or by slicing, thus requiring $T \geq Ln_t$
or $n_t \geq LT$.

\subsection{Existing Rate--Diversity Optimal Constructions}

In \cite{Lu-Kumar},
Lu and Kumar provide for $L = 1$ and $T \geq n_t$
a rate--diversity optimal family for any optimal 
rate--diversity pair. 
The construction is based on a mapping which takes a collection of
maximum rank distance codes over finite fields, 
namely Gabidulin codes \cite{Gabidulin-Codes}, 
to a space--time code which inherits the rank of differences properties 
of the underlying finite field codes. In the case of 
$L > 1$ and $T \geq Ln_t$,
Lu and Kumar
further provide a rate--diversity optimal 
family for any optimal 
rate--diversity pair by horizontal slicing 
as described in Proposition \ref{Horizontal-Slicing-Construction}. 
We digress briefly to outline the connection 
to this paper. Looking at the slicing as 
being done for the underlying finite field code,
the need for $T \geq Ln_t$ occurs precisely due to a 
limitation of rank-metric or Gabidulin codes which is overcome
by sum-rank or linearized Reed--Solomon codes \cite{LRS-Codes}.
Once linearized Reed--Solomon codes are at hand, 
tools in the literature for obtaining space--time codes 
from codes over finite fields can be adopted and a 
multiblock rate--diversity optimal family allowing for
$T < Ln_t$ will follow. Analogous results hold for the case of 
$T < n_t$ and $n_t < LT$ by applying 
the appropriate matrix transpositions. 

Single-block rate--diversity optimal families are also 
described in \cite{ST-Gaussian-Integers,Gabidulin-Space-Time,
	Sven,Lu-AM-PSK,Arbitrary,
	Hammons-Rate-Diversity-Optimal,Rate-Diversity-Tradeoff-Lu}.
These are again based on starting with rank-metric or 
Gabidulin codes and using 
different 
mappings from finite fields to constellations
that are rank-metric-preserving in some sense. 
By Propositions \ref{Repetition-Construction}, \ref{Horizontal-Slicing-Construction}, 
and \ref{Vertical-Slicing-Construction}, these can be used readily to obtain 
multiblock rate--diversity optimal codes in the three special cases 
to which they pertain.
Moreover, a large number of single-block constructions are available in 
the literature some of which may be rate--diversity optimal for some specific 
points on the tradeoff curve, usually the point of full diversity. 
They can therefore potentially be used to construct rate--diversity optimal 
multiblock codes in the aforementioned special cases. 
The reader is referred to the summary of prior constructions in
\cite{Lu-Kumar} 
for details. 
Nonetheless, our focus shall be the unexplored case of 
$n_t \leq T < Ln_t$ and $1 < d < Ln_t$ where the 
methods for adapting single-block constructions fail.

Additionally, other lines of work provide non-explicit constructions
of space--time codes via design criteria that are more
amenable to algebraic constructions or computer search constructions. 
In \cite{PSK-Criterion} and \cite{QAM-Criterion}, translations
of the rank distance over complex field criterion of
\eqref{Complex-Sum-Rank-Distance}
to rank criteria over finite fields or finite rings 
are considered in the single-block case. 
In \cite{Rate-Diversity-Tradeoff-General}, the work of 
\cite{PSK-Criterion} is extended to provide algebraic criteria
for the design of rate--diversity optimal multiblock
codes with BPSK and QPSK constellations
and a few examples found by exhaustive or empirical
searches are provided.

\section{Other Perspectives}\label{Other-Perspectives-Sec}

In this section, we examine the relevance
of the rate--diversity perspective and provide
a framework for comparing the codes
to be constructed with other codes in the space--time
coding literature.
For the purposes of comparing different
space--time codes constructed in different 
manners and designed for different criteria,
we will define some more meaningful notions of rate.
The \textit{bits per channel use (bpcu) rate $R_\mathsf{b}$} 
of an $L$-block $n_t \times T$ space--time code 
$\mathcal{X}$ is defined by
\begin{equation}\label{bpcu-Rate-Definition}
R_\mathsf{b} = \frac{1}{LT}\log_2 \lvert\mathcal{X}\rvert \text{.}
\end{equation}
The \textit{bits per channel use per transmit antenna (bpcu/tx) rate 
	$R_\mathsf{b/tx}$} 
of an $L$-block $n_t \times T$ space--time code 
$\mathcal{X}$ is defined by
\begin{equation}\label{bpcu-tx-Rate-Definition}
R_\mathsf{b/tx} = \frac{1}{n_tLT}\log_2 \lvert\mathcal{X}\rvert 
= \frac{R_\mathsf{b}}{n_t}\text{.}
\end{equation}

\subsection{Unconstrained Transmission Alphabets}

An $L$-block $n_t\times T$ space--time code $\mathcal{X}$ 
is said to be a \textit{linear dispersion code} \cite{Linear-Dispersion}\footnote{We have slightly modified the definition from that of \cite{Linear-Dispersion}.}
if it can be expressed as 
\begin{equation}\label{LD-Code-Definition}
\mathcal{X}
= \left\{ \sum_{i = 1}^{n_tLT} a_i A_i  \mmiddle|
a_1,a_2,\dots,a_{n_tLT} \in \mathcal{A}_\mathsf{in} \right\}
\end{equation}
where $A_1,A_2,\dots,A_{n_tLT}\in \mathbb{C}^{n_t\times LT}$ 
are referred to as \textit{dispersion matrices} and $\mathcal{A}_\mathsf{in}$
is a constellation which we refer to as the 
\textit{input constellation}.
The significance of such codes is that the detection problem 
\eqref{ML-Estimate} can be converted into
an equivalent standard MIMO detection problem.
In particular, one can easily show that the channel \eqref{Channel} can be converted into one
with an $n_tLT \times 1$ transmitted vector with entries from the input
constellation, an $n_rLT \times n_tLT$ effective channel matrix
obtained as a function of the channel matrices and the dispersion matrices,
and an $n_rLT \times 1$ received vector.
This allows them to be decoded using the same methods
used for ML MIMO detection (see, e.g., \cite{Fifty-Year-MIMO}), 
most notably sphere decoding \cite{Sphere-Decoder,Sphere-Decoder-2,ML-CLPS}. More generally,
a wide class of sequential decoding algorithms become readily applicable
\cite{Tree-Search-Decoding}. The codes to be introduced
in this paper are \textit{not} linear dispersion code
and their decoding will be the subject of Section \ref{Decoding-Sec}.

Note that a linear dispersion code is not 
completely over the input constellation 
$\mathcal{A}_\mathsf{in}$.
The codeword entries are linear combinations 
of symbols from $\mathcal{A}_\mathsf{in}$
and thus belong to a larger constellation. This constellation
is usually not of any concern in the literature dealing with such
codes and is sometimes 
referred to as being \textit{unconstrained.} 
However, such language is merely an artifact
of the code construction method 
and nothing prevents us from analyzing these codes from 
a rate--diversity tradeoff perspective. 
Every space--time code is completely over some constellation.
The smallest such constellation is the union of the entries
of all the codewords
\begin{equation*}
\mathcal{A} = 
\bigcup_{X \in \mathcal{X}}
\left\{(X)_{ij}\mmiddle| i\in\{1,2,\dots, n_t\},
j\in\{1,2,\dots, LT\}\right\}\text{.}
\end{equation*}
and can accordingly be used to define the rate.
Linear dispersion codes are usually constructed to be
full diversity in which case the rate would satisfy
\begin{equation}
\label{Rate-Diversity-Tradeoff-Full-Diversity}
R = \frac{R_\mathsf{b}}{
	\log_2\abs{\mathcal{A}}}
\leq \frac{1}{L}\cdot\max\mleft\{\frac{n_t}{T},1\mright\}
\end{equation}
and the code would be rate--diversity optimal if 
and only if this held with equality.

Moreover, linear dispersion codes can be constructed  
to both be full diversity and to have $\abs{\mathcal{A_\mathsf{in}}}^{n_tLT}$ 
distinct codewords. Examples include 
certain codes constructed from CDAs \cite{Division-Algebras} as well
as all of the codes in \cite{Hsiao-Code,Sheng,Perfect-Codes,Improved4}
that will be used as 
empirical error performance baselines in 
this paper.
Regardless of whether or not we have rate--diversity optimality,
we will have a bpcu rate of 
\begin{equation*}
R_\mathsf{b} = n_t \cdot \log_2\abs{\mathcal{A}_\mathsf{in}}\text{.}
\end{equation*}

Furthermore, such codes can be constructed with the input
constellation $\mathcal{A}_\mathsf{in}$ 
being arbitrarily large. Thus, they can provide full diversity
with an arbitrarily large bpcu rate. Another notable
example of codes allowing for this is the codes of
\cite{Universal-ST}.

Given that such codes exist, there is no
tradeoff between bpcu rate and diversity
if the constellation size is not constrained. 
Indeed, if our \textit{only} interest
is in maximizing bpcu rate and diversity, 
the rate--diversity tradeoff is irrelevant 
since we can always impose full diversity.
In such a setting, it is typical to design
for  \textit{diversity--multiplexing tradeoff} 
\cite{DMG} optimality. This typically 
amounts to
asking slightly more of a full diversity linear dispersion code.
In particular, in the case of minimal delay 
linear dispersion codes constructed
from CDAs, it suffices to impose that the code
has a \textit{non-vanishing determinant} property
\cite{NVD,Golden-Code}. This amounts to being 
able to bound the magnitude of the
determinant of the codeword differences (or product of the 
determinants of sub-codeword differences 
in the multiblock case) 
away from zero independently of $R_\mathsf{b}$.

The codes to be introduced in this paper 
are neither designed 
for diversity--multiplexing optimality nor 
naturally amenable to an analysis of this tradeoff.
However, this
is not to suggest that it is not possible or that
we are dealing with a fundamentally different kind of
a code. 
For example, in \cite{DM-RANK}, bounds on the diversity--multiplexing
tradeoff for similarly constructed codes are obtained. 
Nonetheless, the only relevance of the 
diversity--multiplexing tradeoff to this paper 
is that the codes to be used as performance baselines
in this paper happen to be optimal with respect
to this tradeoff. The baseline codes have been
chosen based on the fact (to the best
of the authors' knowledge) that they are
the only explicitly described
 \textit{multiblock} codes
admitting a feasible decoder. 
In any case, such codes
are standard benchmarks in the single-block 
setting as well.

We will now consider two situations in which
the rate--diversity tradeoff might be relevant. 
In particular, we will consider
situations where we might
be interested in codes that are not full diversity. 
These situations are:
\begin{itemize}
	\item The size or nature of the constellation is of concern.
	\item The low-SNR error performance is of concern.
\end{itemize}
Importantly, we argue that when a space--time coded system is
scaled in a natural way, both of these issues necessarily 
become of practical concern.

We begin with the first situation.
It is well-known that large constellations
are associated with 
implementation challenges. 
For example, \cite{Robustness,Integer-Codes} consider
transmitter-side quantization as well as peak-to-average power 
ratio issues arising from the constellations
produced by certain
linear dispersion codes.
Receiver-side quantization issues which are exacerbated
by large constellations are studied
in \cite{Finite-Precision,Quantized-MIMO}. Moreover, in \cite{Hsiao-Signalling-Complexity},
techniques for reducing the signalling complexity, i.e.,
constellation size, for CDA-based linear dispersion codes are provided.
However, as will be seen in the next section, such an issue
is inherent to full diversity codes.

\subsection{A Signalling-Complexity Perspective}\label{sig-comp-subsection}

Consider an $L$-block $n_t\times T$ space--time code completely over
some constellation $\mathcal{A}$ with rate $R$ and transmit
diversity gain of $d$ satisfying $1 \leq d \leq L \cdot \min\{n_t,T\}$.
Next, take the diversity gain to be a $(1-\varepsilon)$
fraction of the total available diversity gain. In particular,
fix some $\varepsilon$ satisfying $0 \leq \varepsilon < 1$
and let $d = \ceil*{(1-\varepsilon)\cdot L\cdot\min\{n_t,T\}}$.
Consider further fixing $R_\mathsf{b/tx} = R_\mathsf{b}/n_t$.

Multiplying both sides of
the rate--diversity tradeoff 
\eqref{Rate-Diversity-Tradeoff-Inequality} 
by $1/n_t$, we get
\begin{align*}
\frac{R_\mathsf{b/tx}}{\log_2\abs{\mathcal{A}}}
&\leq  1 - 
\frac{\ceil*{(1-\varepsilon)\cdot
		L\cdot\min\{n_t,T\}}}{L\cdot\min\{n_t,T\}}
+ \frac{1}{L\cdot\min\{n_t,T\}} \\
&\leq \varepsilon
+ \frac{1}{L\cdot\min\{n_t,T\}}
\end{align*}
hence
\begin{equation*}
\abs{\mathcal{A}} 
\geq \exp(\frac{R_\mathsf{b/tx}\ln 2}{\varepsilon + \frac{1}{ L\cdot\min\{n_t,T\}}})\text{.}
\end{equation*}
Thus, we have
a lower bound on the constellation size which we will
view as a function $F_\varepsilon$ of $L\cdot\min\{n_t,T\}$, i.e.,
define
\begin{equation}\label{Constel-Size-LB-Function}
F_\varepsilon(L\cdot\min\{n_t,T\}) = \exp(\frac{R_\mathsf{b/tx}\ln 2}{\varepsilon + \frac{1}{ L\cdot\min\{n_t,T\}}})
\text{.}
\end{equation}
This gives the constellation size obtained by
a rate--diversity optimal code 
(provided that 
$(1-\varepsilon)\cdot L\cdot\min\{n_t,T\}$ is an integer)
and is also the smallest constellation size possible.
\begin{figure}[t]
	\centering
	\includegraphics[width=\columnwidth]
	{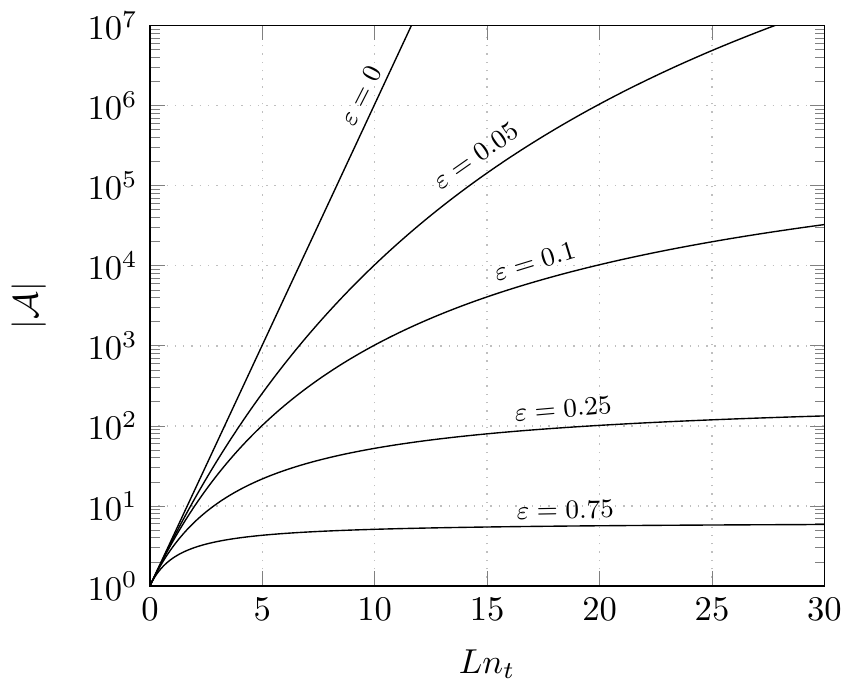}
	\caption{Constellation size lower bounds for
		$T\geq n_t$, $R_\mathsf{b}/n_t = 2$, 
		and $d =\ceil*{(1-\varepsilon)Ln_t}$}\label{Constellation-Size-Bounds}
\end{figure}

Observe that for $\varepsilon = 0$, we have 
\begin{equation}
F_0(L\cdot\min\{n_t,T\}) = \exp(L\cdot\min\{n_t,T\} \cdot R_\mathsf{b/tx}\ln 2)\label{Exponential-Constellation}
\end{equation}
and for $\varepsilon > 0$, we have 
\begin{equation}
F_\varepsilon(L\cdot\min\{n_t,T\}) < \exp(\frac{R_\mathsf{b/tx}\ln 2}{\varepsilon})\label{Bounded-Constellation}\text{.}
\end{equation}
We now consider scaling the system in $L\cdot\min\{n_t,T\}$.
This has two
very natural interpretations:

The first is to take $n_t$ and $T$ to be fixed
and $L$ to be growing. In this case, a fixed $R_\mathsf{b/tx}$ corresponds
to a fixed $R_\mathsf{b}$. Thus, we are scaling the system in the 
number of fading blocks $L$ being coded across with a fixed bpcu rate $R_\mathsf{b}$
and improving reliability in $(1-\varepsilon)\cdot L$. 
We then see that we have exponential growth of the constellation
size in $L$ for $\varepsilon = 0$ and a constellation size that is 
bounded independently of $L$ at best for $\varepsilon > 0$.

The second is to take $L$ to be fixed and $n_t$ to be growing
with $T \geq n_t$. The $T \geq n_t$ requirement is needed to
obtain the transmit diversity gain afforded by $n_t$ transmit antennas
and is standard in diversity--multiplexing literature. Moreover,
the bpcu rate $R_\mathsf{b} = R_\mathsf{b/tx}\cdot n_t$ grows linearly 
with $n_t$ so as to match the rate of uncoded independent signalling
on each antenna. This is also as would be the
case for families of diversity--multiplexing optimal
linear dispersion codes.
Again, we have exponential growth
of the constellation size in $n_t$ for $\varepsilon = 0$ 
and potentially bounded constellation size otherwise.

Fig.~\ref{Constellation-Size-Bounds} plots
constellation size lower bounds for $R_\mathsf{b/tx} = 2$
and $T \geq n_t$. They can be interpreted from either
of the points of view just discussed. From this, we
see that full diversity codes necessitate exponential
growth in the constellation size. On the other hand, 
by giving up only a fixed fraction of the total
available diversity gain, a bounded
constellation size is possible\footnote{In actuality,
	the code construction will require $\abs{\mathcal{A}} > L$
	(see Definition \ref{LRS-Code-Definition} in Section
	\ref{LRS-Def-Sec}). Whether
	this bound is fundamental is connected to some open
	problems. The reader is referred to the relevant
	discussions at the end of \cite{Umberto-Multishot}.}.

\subsection{Low-SNR Error Performance}

Recall that diversity gain is only a 
characterization of the error performance
in the limit of large SNR. 
In \cite{Vucetic-Ch2}, it is argued that a
full diversity requirement 
becomes increasingly less important as
the total available diversity gain increases.
Analytical arguments and simulation evidence are
provided with trellis codes and in a single-block setting.
In particular, it may be the case that a code 
which is not full diversity
code has relatively few codeword pairs for which
the sum of ranks \eqref{Complex-Sum-Rank-Distance}
is minimal. 
With that said, we will not attempt to analyze the distance distributions
of the proposed codes nor connect them to rate--diversity
optimality. 
Rather,
we will simply demonstrate in simulation
that the proposed codes can have good low-SNR
performance. 

Finally, we remark that what
constitutes a low SNR is relative and one need
not be nominally interested in low-SNR performance
to back away from full diversity codes. We can
replace the term \textit{low} with \textit{finite}
in the previous discussion and it will hold. This
will also be reflected by the simulation results.

\section{Code Construction}\label{Code-Construction-Sec}

This section will begin with the required
background on sum-rank codes after which
we define linearized Reed--Solomon codes
in Section \ref{LRS-Def-Sec}. 
Following this, we provide some background
and results regarding rank-metric-preserving
maps in Section \ref{RMPM-Sec} which
are applied in Section \ref{Construction-Sec}
to yield the main code construction.

\subsection{Background on Sum-Rank Codes}
\label{Sum-Rank-Codes-Section}

The sum-rank metric was introduced 
in the context of multishot network coding in \cite{Sum-Rank-Metric}.
We shall define it adopting the more recent 
framework of \cite{Umberto-Multishot,Umberto-MR-LRC}.
The space--time coding analogue of the sum-rank metric 
occurring in \eqref{Complex-Sum-Rank-Distance} appears earlier 
in \cite{Rate-Diversity-Tradeoff-General,Special-Case-Conf,Lu-Kumar}. 
We will use the terms \textit{metric} and \textit{distance} interchangeably.

The finite field with $q$ elements
(with $q$ a prime power)
will be denoted by $\mathbb{F}_q$. 
We adopt the convention that 
$\mathbb{F}_q^s = \mathbb{F}_q^{1 \times s}$.

The extension field $\mathbb{F}_{q^m}$ 
forms an $m$-dimensional 
vector space over the base field $\mathbb{F}_q$. 
Therefore, there exists a vector space isomorphism 
between $\mathbb{F}_{q^m}$ and $\mathbb{F}_{q}^m$.
We will accordingly define a matrix representation 
of the elements of $\mathbb{F}_{q^m}^s$ with respect
to some choice of basis. 
Let $\mathcal{B} = (\beta_1, \beta_2, \dots, \beta_m)$ 
be an ordered basis of $\mathbb{F}_{q^m}$ over $\mathbb{F}_q$. 
Then, any $\mathbf{c} \in \mathbb{F}_{q^m}^s$ can be written as 
\begin{align*}
\mathbf c 
& =  
\begin{bmatrix}
\sum_{i=1}^m \beta_i c_{i1} 
& \sum_{i=1}^m \beta_i c_{i2} 
& \cdots 
& \sum_{i=1}^m \beta_i c_{is}
\end{bmatrix}\\
&= 
\sum_{i=1}^m \beta_i 
\begin{bmatrix} c_{i1} & c_{i2} & \cdots & c_{is} \end{bmatrix}\\
&= \sum_{i=1}^m \beta_i \mathbf{c}_i
\end{align*}
where $\mathbf{c}_i = \begin{bmatrix} c_{i1} & c_{i2} & \cdots & c_{is}
\end{bmatrix} \in \mathbb{F}_{q}^s$ contains the $i$th coordinates of
the representations of the $s$ elements 
of $\mathbb{F}_{q^m}$ in $\mathbf{c}$ in terms of 
the basis $\mathcal B$ for $i=1,2,\dots,m$.

\begin{definition}[Matrix Representation Map \cite{Umberto-MR-LRC}]
	Let $\mathcal{B} = (\beta_1, \beta_2, \dots, \beta_m)$ 
	be an ordered basis of $\mathbb{F}_{q^m}$ over $\mathbb{F}_q$.
	The matrix representation of any 
	$\mathbf{c} \in \mathbb{F}_{q^m}^s$ 
	with respect to $\mathcal B$ is given by the 
	\textit{matrix representation map}
	${M_\mathcal{B}\colon\mathbb{F}_{q^m}^s 
		\longrightarrow \mathbb{F}_{q}^{m\times s}}$ defined by
	\begin{equation*}
	M_\mathcal{B}\mleft(\sum_{i=1}^m \beta_i\mathbf{c}_i\mright) = 
	\begin{bmatrix}
	c_{11} & c_{12} & \cdots & c_{1s} \\
	c_{21} & c_{22} & \cdots & c_{2s} \\		
	\vdots & \vdots & \ddots & \vdots \\		
	c_{m1} & c_{m2} & \cdots & c_{ms} \\				
	\end{bmatrix}
	\end{equation*}
	where $\mathbf{c} = \sum_{i=1}^m \beta_i \mathbf{c}_i$
	and
	$\mathbf{c}_i = \begin{bmatrix} c_{i1} & c_{i2} & \cdots & c_{is}
	\end{bmatrix} \in \mathbb{F}_{q}^s$ for $i=1,2,\dots,m$.
\end{definition}

	For some fixed positive integer $N$, we define a 
	\mbox{\textit{sum-rank length partition} \cite{Umberto-MR-LRC}}
	as a choice of a positive integer $L$
	and ordered positive integers $r_1,r_2,\dots,r_L$ such that 
	${N = r_1 + r_2 + \cdots + r_L}$. For a fixed sum-rank length partition, 
	any $\mathbf{c} \in \mathbb{F}_{q^m}^N$ can be partitioned as 
	$$\mathbf c = \begin{bmatrix} \mathbf{c}^{(1)} & 
	\mathbf{c}^{(2)} & \cdots & \mathbf{c}^{(L)} \end{bmatrix}$$ 
	where $\mathbf{c}^{(\ell)} \in \mathbb{F}_{q^m}^{r_\ell}$ for 
	$\ell = 1,2,\dots,L$. Moreover,
	it shall be understood that any 
	$\mathbf{c} \in \mathbb{F}_{q^m}^N$ 
	partitions this way once a sum-rank length partition has been
	specified.

\begin{remark*}
	\begin{sloppypar}
		The matrix representation map is $\mathbb{F}_q$-linear.
		In particular, for any ${A,B \in \mathbb{F}_q^{s\times t}}$ and 
		$\mathbf{c},\mathbf{d} \in \mathbb{F}_{q^m}^s$, we have
		$$
		M_\mathcal{B}(\mathbf c A + \mathbf d B)
		= M_\mathcal{B}(\mathbf c) A + M_\mathcal{B}(\mathbf d) B\text{.}
		$$
		Furthermore, for any $a,b \in \mathbb{F}_q$ and 
		$\mathbf{c},\mathbf{d} \in \mathbb{F}_{q^m}^s$, we have
		$$
		M_\mathcal{B}(a\mathbf c  + b\mathbf d )
		= aM_\mathcal{B}(\mathbf c) + bM_\mathcal{B}(\mathbf d)\text{.}
		$$
	\end{sloppypar}
\end{remark*}

\begin{definition}[Sum-Rank Metric \cite{Sum-Rank-Metric,
		Umberto-Multishot,Umberto-MR-LRC}]
		Let ${\mathcal{B} = (\beta_1, \beta_2, \dots, \beta_m)}$ be an ordered
		basis of $\mathbb{F}_{q^m}$ 
		over $\mathbb{F}_q$
		and fix 
		a sum-rank length partition $N = r_1 + r_2 + \cdots + r_L$ so that 
		any $\mathbf{c} \in \mathbb{F}_{q^m}^N$ partitions as
		$\mathbf c = \begin{bmatrix} \mathbf{c}^{(1)} & 
		\mathbf{c}^{(2)} & \cdots & \mathbf{c}^{(L)} \end{bmatrix}$ 
		with $\mathbf{c}^{(\ell)} \in \mathbb{F}_{q^m}^{r_\ell}$ for 
		$\ell = 1,2,\dots,L$. 
		The \textit{sum-rank weight} is
		the function
		$w_{\mathsf{SR}}\colon \mathbb{F}_{q^m}^N \longrightarrow \mathbb{N}$  
		defined by
		\begin{equation*}
		w_\mathsf{SR}(\mathbf{c}) = \sum_{\ell=1}^L
		\rank(M_\mathcal{B}(\mathbf{c}^{(\ell)}))
		\end{equation*}
		for any $\mathbf{c} \in \mathbb{F}_{q^m}^N$. 
		The 
		\textit{sum-rank metric} 
		(or distance) is the function
		${d_{\mathsf{SR}}\colon \mathbb{F}_{q^m}^N \times\mathbb{F}_{q^m}^N
			\longrightarrow \mathbb{N}}$ defined by
		\begin{align*}
		d_\mathsf{SR}(\mathbf{c},\mathbf{d}) 
		&= w_\mathsf{SR}(\mathbf{c}-\mathbf{d})\\
		&= \sum_{\ell=1}^L
		\rank(M_\mathcal{B}((\mathbf{c}-\mathbf{d})^{(\ell)}))\\
		&= 	\sum_{\ell=1}^L
		\rank(M_\mathcal{B}(\mathbf{c}^{(\ell)})-
		M_\mathcal{B}(\mathbf{d}^{(\ell)}))
		\text{.}
		\end{align*}
		for any $\mathbf{c},\mathbf{d}\in\mathbb{F}_{q^m}^N$.
\end{definition}
\begin{sloppypar}
	In what follows, unless otherwise stated,
	we shall assume some fixed but arbitrary
	sum-rank length partition 
	${N = r_1 + r_2 + \cdots + r_L}$ for some fixed positive integer
	$N$ for the purposes of defining sum-rank distances and
	weights. The reader should keep in mind the dependency of
	these quantities
	on the sum-rank length partition which we
	suppress for brevity.
	On the other hand, we need not specify a basis since the rank
	of a matrix representation is independent of the choice of basis.
\end{sloppypar}

\begin{remark*}
	\label{Sum-Rank-Metric-Is-A-Metric}
	The subadditivity of rank, i.e., 
	that for a pair of
	equal-sized matrices $A$ and $B$ over any field, we have 
	$\rank(A+B) \leq \rank(A) + \rank(B)$ can
	be used to verify that the sum-rank metric
	is indeed a metric.
\end{remark*}

\begin{sloppypar}
	Observe that when the sum-rank length partition is 
	${r_1 = r_2 = \cdots = r_N = 1}$
	(with $L = N$), 
	the sum-rank distance recovers 
	\textit{Hamming distance} \cite{Hamming} 
	since the zero column vector is the only column vector of rank zero. 
	We denote Hamming distance by $d_\mathsf{H}$. 
	At the other extreme of sum-rank length partition
	$N = r_1$ (with $L = 1$), 
	the sum-rank distance becomes the \textit{rank distance} \cite{Gabidulin-Codes}
	which we denote by $d_\mathsf{R}$. In this sense,
	the sum-rank metric generalizes the Hamming and rank metrics.
	Our particular notion of rank distance, i.e.,
	as a metric on vector spaces over extension fields,
	was first studied by Gabidulin in \cite{Gabidulin-Codes}.
	A related notion of rank distance occurs earlier in \cite{Rank-Origin}.
\end{sloppypar}

We define the \textit{minimum sum-rank distance} 
of a code ${\mathcal C \subseteq \mathbb{F}_{q^m}^N}$, denoted
$d_\mathsf{SR}(\mathcal{C})$, as
\begin{align*}\label{Minimum-Sum-Rank-Distance}
d_\mathsf{SR}(\mathcal{C})&= 
\min_{\substack{\mathbf{c},\mathbf{d} \in \mathcal{C}
		\\\mathbf{c}\neq\mathbf{d}}}
d_\mathsf{SR}(\mathbf{c},\mathbf{d})\\
&= \min_{\substack{\mathbf{c},\mathbf{d} \in \mathcal{C}
		\\\mathbf{c}\neq\mathbf{d}}}
\sum_{\ell=1}^L
\rank(M_\mathcal{B}(\mathbf{c}^{(\ell)})-
M_\mathcal{B}(\mathbf{d}^{(\ell)}))
\text{.}
\end{align*}
A code $\mathcal{C}\subseteq\mathbb{F}_{q^m}^N$is said
to be \textit{linear} if it is a $k$-dimensional subspace
of the vector space $\mathbb{F}_{q^m}^N$ over the field
$\mathbb{F}_{q^m}$. Adopting the standard concise notation,
such a code will be said to be an $[N,k]_{q^m}$ code
where the square brackets emphasize the linearity. 

Theorem \ref{Singleton-General-Theorem}, which follows, is
analogous to Theorem \ref{Rate-Diversity-Tradeoff-Theorem},
the rate--diversity tradeoff,
and may be proven in a similar manner.
\begin{theorem}[A Generalized Singleton Bound 
	\cite{LRS-Codes,Umberto-MR-LRC}]
	\label{Singleton-General-Theorem}
	For any code $\mathcal C \subseteq \mathbb{F}_{q^m}^N$, we have
	\begin{equation}\label{Singleton-General}
	\vert \mathcal{C} \vert \leq q^{m(N-d_\mathsf{SR}(\mathcal{C})+1)}\text{.}
	\end{equation}
	Moreover, if $\mathcal C$ is an $[N,k]_{q^m}$ code, 
	this is equivalent to
	\begin{equation}
	k \leq N - d_\mathsf{SR}(\mathcal{C}) + 1\text{.}
	\end{equation}
\end{theorem}

Reiterating, it is implicit that all claims made about sum-rank
distance are true
for a fixed but arbitrary choice of sum-rank length partition.
Codes which achieve \eqref{Singleton-General} with equality are 
said to be \textit{maximum sum-rank distance (MSRD).} 
When the sum-rank length partition
is chosen so that the sum-rank metric becomes the 
Hamming metric, 
\eqref{Singleton-General} becomes the classical 
Singleton bound \cite{Singleton} 
achieved by classical Reed--Solomon codes \cite{RS-Codes}.
MSRD thus becomes equivalent to
\textit{maximum distance separable (MDS)}.
When the sum-rank length partition 
is chosen so that the sum-rank metric becomes the
rank metric, being MSRD becomes equivalent to being 
\textit{maximum rank distance (MRD)} 
which is achieved by Gabidulin codes \cite{Gabidulin-Codes}.

By 
\eqref{Finer-Partitions-Preserve-Total-Rank}, it is easy
to see that for any $\mathbf{c},\mathbf{d}\in \mathbb{F}_{q^m}^N$,
we have
\begin{equation}
d_\mathsf{R}(\mathbf{c},\mathbf{d}) \leq 
d_\mathsf{SR}(\mathbf{c},\mathbf{d}) \leq 
d_\mathsf{H}(\mathbf{c},\mathbf{d})\text{.}
\end{equation}
One can further see that for any
$\mathcal{C}\subseteq\mathbb{F}_{q^m}^N$, we have
\begin{equation}\label{Distance-Strengths}
d_\mathsf{R}(\mathcal{C}) \leq d_\mathsf{SR}(\mathcal{C}) 
\leq d_\mathsf{H}(\mathcal C)\text{.}
\end{equation}
Straightforwardly generalizing this, we get the following result
from \cite{Umberto-MR-LRC}:

\begin{proposition}[Refinement Preserves Sum-Rank \cite{Umberto-MR-LRC}]
	\label{Finer-Partitions-Preserve-Total-Rank-Finite-Field}
	Denote by $d_\mathsf{SR}^\mathsf{coarse}$
	the sum-rank metric for the sum-rank
	length partition $N = \sum_{\ell=1}^L r_\ell$.
	Let $r_\ell = \sum_{i=1}^{L_\ell} r_{\ell i}$
	for $\ell = 1,2,\dots,L$.
	Denote by $d_\mathsf{SR}^\mathsf{fine}$
	the sum-rank metric for the sum-rank length
	partition $N = \sum_{\ell=1}^L\sum_{i=1}^{L_\ell} r_{\ell i}$.
	For any code $\mathcal{C}\subseteq \mathbb{F}_{q^m}^N$, we have
	that
	\begin{equation}
	d_\mathsf{SR}^\mathsf{coarse}(\mathcal{C})
	\leq
	d_\mathsf{SR}^\mathsf{fine}(\mathcal{C})\text{.}
	\end{equation}
\end{proposition}

A crucial aspect of Theorem \ref{Singleton-General-Theorem}
is that we have the same bound regardless of the choice
of sum-rank length partition---coarsening the metric does 
not lead to a codebook size penalty. Theorem 
\ref{Finer-Partitions-Preserve-Total-Rank-Finite-Field} 
thus results in the following:

\begin{corollary}
	\label{Finer-Partitions-Preserve-Total-Rank-Finite-Field-Corollary}
	Let $d_\mathsf{SR}^\mathsf{coarse}$ 
	and 
	$d_\mathsf{SR}^\mathsf{fine}$ be as
	previously defined.
	If $\mathcal{C}\subseteq \mathbb{F}_{q^m}^N$ 
	is MSRD with respect to $d_\mathsf{SR}^\mathsf{coarse}$,
	then it is also MSRD with respect to 
	$d_\mathsf{SR}^\mathsf{fine}$.
\end{corollary}

As special cases of this corollary, we have the following:
\begin{itemize}
	\item If $\mathcal C \subseteq \mathbb{F}_{q^m}^N$ is MRD,
	then $\mathcal{C}$ is MSRD for \textit{any} sum-rank length partition.
	\item 	If $\mathcal C \subseteq \mathbb{F}_{q^m}^N$ is MSRD, 
	then $\mathcal{C}$ is MDS.
\end{itemize}

So, if $\mathcal{C} \subseteq \mathbb{F}_{q^m}^N$ is MRD, then 
$d_\mathsf{R}(\mathcal{C}) = d_\mathsf{SR}(\mathcal{C})
= d_\mathsf{H}(\mathcal{C}) = N-\log_{q^m}\abs{\mathcal{C}}+1$
making it MSRD and MDS.
In light of this, it is natural to ask 
why one might use an MSRD code over an MRD code. 
The
answer lies in the following results from 
\cite{Umberto-MR-LRC}:

\begin{theorem}[Extension Degree Bound
	\cite{Umberto-MR-LRC}]\label{Extension-Degree-Bound}
	Let ${\mathcal C \subseteq \mathbb{F}_{q^m}^N}$ 
	be an MSRD code for the sum-rank length partition 
	${r_1 = r_2 = \cdots = r_L = N/L}$ and let 
	$d_\mathsf{SR}(\mathcal{C}) > 1$, then
	\begin{equation}\label{Extension-Degree-Bound-Inequality}
	m \geq \frac{N}{L}\text{.}
	\end{equation}
\end{theorem}
\begin{corollary}\label{MRD-Extension-Degree-Bound}
	If $\mathcal C \subseteq \mathbb{F}_{q^m}^N$ is MRD 
	and $d_\mathsf{R}(\mathcal{C}) > 1$, then $m \geq N$.
\end{corollary}

Thus, MRD codes impose a larger field extension degree. 
In particular, 
for the purposes of our space--time code construction, in the case of
$T \geq n_t$, we will let $m = T$ and $N = Ln_t$. 
Corollary \ref{MRD-Extension-Degree-Bound} then
becomes the requirement that $T \geq Ln_t$ 
for constructions based on MRD codes
such as those in \cite{Lu-Kumar,Special-Case-Conf}.
In the case
of $T \leq n_t$, we will let $m = n_t$ and $N = LT$ in which case
Corollary \ref{MRD-Extension-Degree-Bound} becomes an
$n_t \geq LT$ requirement for constructions based on MRD codes.

Linearized Reed--Solomon codes, to be introduced
in the next subsection, are a recently introduced \cite{LRS-Codes} 
family of linear MSRD codes which can achieve the field extension degree bound
\eqref{Extension-Degree-Bound-Inequality} with equality. This will enable
multiblock space--time codes constructed from them to achieve $T = n_t$,
or more generally, any relationship between $T$ and $n_t$. 

We conclude this section
with some comments on generator matrices. 
A \textit{generator matrix} 
of an $[N,k]_{q^m}$ code
is a full rank matrix $G \in \mathbb{F}_{q^m}^{k\times N}$
whose row space is the code. A \textit{systematic} 
generator matrix is one of the form
$G = 
\begin{bmatrix}
I_k & P
\end{bmatrix}$
where  
$I_k \in \mathbb{F}_{q^m}^{k\times k}$
denotes the $k\times k$ identity matrix
and 
$P \in \mathbb{F}_{q^m}^{k\times (N-k)}$.
Clearly, 
any linear code has a generator
matrix which is some column permutations
away from a systematic one. While column
permutations change the row space and hence
the code, they are
isometries of the Hamming
and rank metrics making systematicity a non-issue
when these are the metrics of interest. 
However, it is easy to see that 
column permutations are not in general
isometries of the sum-rank metric.

With that said,
this will be a non-issue because for 
any MSRD code, there exists a systematic
generator matrix and, more generally,
a generator matrix with the $k$ columns
of $I_k$ occurring \textit{anywhere}
in it. This follows from the fact
that MSRD codes are MDS and these
are well-known properties of MDS
codes (see, e.g., \cite{Roth-Ch11}).

\subsection{Linearized Reed--Solomon Codes}\label{LRS-Def-Sec}

Linearized Reed--Solomon codes
were recently introduced by 
Mart\'{i}nez-Pe\~{n}as in \cite{LRS-Codes}.
These codes
are closely connected to skew polynomial evaluation codes
\cite{Siyu-Dissertation,Siyu-Skew-Evaluation-Codes,Skew-Other}.
Their significance to us should be clear at this point
so we will limit our task in this section to providing 
the definition.

In particular, we will provide a 
specialization 
of the definition of linearized 
Reed--Solomon codes in \cite{LRS-Codes}
which is sufficient for the purposes of 
our space--time code
construction.

We begin by specializing the sum-rank length partition to 
${r_1 = r_2 = \cdots = r_L = N/L}$
maintaining this henceforth.
Define the field automorphism 
$\sigma\colon\mathbb{F}_{q^m} 
\longrightarrow \mathbb{F}_{q^m}$
by $\sigma(a) = a^q$
for all $a \in \mathbb{F}_{q^m}$.

\begin{definition}[Truncated norm \cite{Umberto-MR-LRC}]
	Define $N_i\colon\mathbb{F}_{q^m}
	\longrightarrow\mathbb{F}_{q^m}$ by
	\begin{equation*}
	N_i(a) = \sigma^{i-1}(a)\sigma^{i-2}(a)\cdots\sigma(a)a\text{.}
	\end{equation*}
	for all	$a \in \mathbb{F}_{q^m}$ and all $i \in \mathbb{N}$.
\end{definition}

\begin{definition}[Linear Operator (Composition)
	\cite{Umberto-MR-LRC,LRS-Codes}]
	For some fixed $a \in \mathbb{F}_{q^m}$, 
	define the $\mathbb{F}_q$-linear operator
	$\mathcal{D}_a^i\colon\mathbb{F}_{q^m}\longrightarrow\mathbb{F}_{q^m}$ 
	by
	\begin{equation*}
	\mathcal{D}_a^i (b) = \sigma^i(b)N_i(a)
	\end{equation*} 
	for all $b \in \mathbb{F}_{q^m}$ and all $i \in \mathbb{N}$.
\end{definition}

\begin{remark*}
	We have $\mathcal{D}_a^1(b) = \sigma(b)a$ and 
	$\mathcal{D}_a^1 \circ \mathcal{D}_a^i = \mathcal{D}_a^{i+1}$
	for all $i \in \mathbb{N}$. 
\end{remark*}

For an $[N,k]_{q^m}$ code
with generator matrix
$G \in\mathbb{F}_{q^m}^{k \times N}$, it shall be understood
that the generator matrix partitions according to the sum-rank
length partition as
\begin{equation*}
G = \begin{bmatrix}G_1 & G_2 
& \cdots & G_L\end{bmatrix}
\end{equation*}
with $G_\ell \in \mathbb{F}_{q^m}^{k \times N/L}$ for $\ell=1,2,\dots,L$ 
referred to as \textit{sub-codeword generators}.

\begin{definition}[Linearized Reed--Solomon Code (Special Case)
	\cite{LRS-Codes}]\label{LRS-Code-Definition}
	\begin{sloppypar}
		
		Let ${q > L}$, ${m \geq N/L}$, $\alpha$ be a primitive element of 
		$\mathbb{F}_{q^m}$, and 
		$\mathcal{B} = (\beta_1, \beta_2, \dots, \beta_m)$ 
		be an ordered basis of $\mathbb{F}_{q^m}$ over $\mathbb{F}_q$.
		A linearized Reed--Solomon code is an $[N,k]_{q^m}$
		code with sub-codeword generators
		\begin{equation*}
		G_\ell = 
		\begin{bmatrix}
		\beta_1 & \beta_2 & \cdots & \beta_{N/L} \\
		
		\mathcal{D}^1_{\alpha^{\ell-1}}(\beta_1) & 
		\mathcal{D}^1_{\alpha^{\ell-1}}(\beta_2) & \cdots &
		\mathcal{D}^1_{\alpha^{\ell-1}}(\beta_{N/L})\\
		\mathcal{D}^2_{\alpha^{\ell-1}}(\beta_1) & 
		\mathcal{D}^2_{\alpha^{\ell-1}}(\beta_2) & \cdots &
		\mathcal{D}^2_{\alpha^{\ell-1}}(\beta_{N/L})\\
		\vdots & \vdots & \ddots & \vdots \\
		\mathcal{D}^{k-1}_{\alpha^{\ell-1}}(\beta_1) &
		\mathcal{D}^{k-1}_{\alpha^{\ell-1}}(\beta_2) & \cdots &
		\mathcal{D}^{k-1}_{\alpha^{\ell-1}}(\beta_{N/L})\\
		\end{bmatrix}
		\end{equation*}
		for $\ell = 1,2,\dots,L$.
	\end{sloppypar}
\end{definition}
\begin{theorem}[MSRD Property \cite{LRS-Codes}]\label{MSRD-Property}
	An $[N,k]_{q^m}$ 
	linearized Reed--Solomon code $\mathcal{C}$ achieves 
	$k = N - d_\mathsf{SR}(\mathcal{C}) + 1$.
\end{theorem}

\subsection{Background on Rank-Metric-Preserving Maps}\label{RMPM-Sec}

A variety of methods
\cite{Lu-Kumar,ST-Gaussian-Integers,Arbitrary,
	Sven,PSK-Criterion,QAM-Criterion,Gauss-Int-Map-Is-RMP,Lu-AM-PSK} 
exist for obtaining space--time codes from codes over finite fields. 
We set aside those focused on binary fields because 
linearized Reed--Solomon codes are only interesting when nonbinary 
due to the $q > L$ requirement. We further set aside those which do
not lead to explicit constructions. This leaves over 
the methods in
\cite{Arbitrary,ST-Gaussian-Integers,Gabidulin-Space-Time,Sven, Lu-Kumar}.
In \cite{Arbitrary}, the authors propose a 
framework which encompasses as special cases mappings to the 
Gaussian \cite{ST-Gaussian-Integers,Gabidulin-Space-Time} 
and Eisenstein \cite{Sven} integers and find that, with the exception of when a 
small PSK constellation is required, the method in \cite{Lu-Kumar} is outperformed.
We accordingly define a notion of rank-metric-preserving map only general enough to
subsume these important special cases. 

\begin{definition}[Rank-Metric-Preserving Map]
	Let $q$ be a prime power, $\mathcal{A}$
	be a constellation of size $q$, and  
	$\phi\colon \mathbb{F}_{q} \longrightarrow \mathcal{A}$
	be a bijection.
	Define 
	$\tilde\phi\colon\mathbb{F}_{q}^{n_t\times T} 
	\longrightarrow \mathcal{A}^{n_t\times T}$ 
	to be the corresponding entrywise map
	which is the bijection defined,
	for all $C \in \mathbb{F}_q^{n_t \times T}$, by 
	$(\tilde\phi(C))_{ij} = \phi((C)_{ij})$ for 
	$i = 1,2,\dots,n_t$ and
	$j = 1,2,\dots, T$.
	The map $\phi$ is said to be \textit{rank-metric-preserving} 
	if
	\begin{equation*}
	\rank(\tilde\phi(C)-\tilde\phi(D)) \geq \rank(C-D)
	\end{equation*}
	for all 
	$C,D \in \mathbb{F}_{q}^{n_t\times T},\,C\neq D$.
\end{definition}

We will next introduce 
the Gaussian and Eisenstein 
integers and recall some
useful properties. Some will be
essential to the encoding
and decoding procedures for
the proposed codes and others
will be relevant to answering
existential questions.

In what follows, 
$\imath$ will denote the imaginary unit and $\omega$
will denote the primitive cube root of unity
\begin{equation*}
\omega = \exp(\frac{2\pi\imath}{3}) = -\frac{1}{2} 
+ \imath\frac{\sqrt{3}}{2}\text{.}
\end{equation*}

The \textit{Gaussian integers} which we denote by $\mathbb{G}$ 
are defined by \cite{Conway-Sloane-Ch2}
\begin{equation*}
\mathbb{G} = \mathbb{Z}[\imath] = \{a+b\imath \mid a,b\in \mathbb{Z}\}\text{.}
\end{equation*}
Noting that as a real vector space, $\mathbb{C}$ is isomorphic to $\mathbb{R}^2$, we see that, additively, $\mathbb{G}$ is 
isomorphic to the square lattice ${\mathbb{Z}^2 \subseteq \mathbb{R}^2}$, 
i.e., the set of integral linear combinations of
$\begin{bmatrix} 1 & 0 \end{bmatrix}^\intercal$ and 
$\begin{bmatrix} 0 & 1 \end{bmatrix}^\intercal$.
The \textit{Eisenstein integers} which we denote by $\mathbb{E}$ are defined by
\cite{Conway-Sloane-Ch2}
\begin{equation*}
\mathbb{E} = \mathbb{Z}[\omega] = \{a+b\omega \mid a,b\in \mathbb{Z}\}\text{.}
\end{equation*}
Additively, $\mathbb{E}$ is 
isomorphic to the hexagonal lattice in $\mathbb{R}^2$ spanned by the integral 
linear combinations of
$\begin{bmatrix} 1 & 0 \end{bmatrix}^\intercal$ and 
$\begin{bmatrix} -1/2 & \sqrt{3}/2 \end{bmatrix}^\intercal$.

\begin{figure}[t]
	\centering
	\includegraphics[width=\columnwidth]{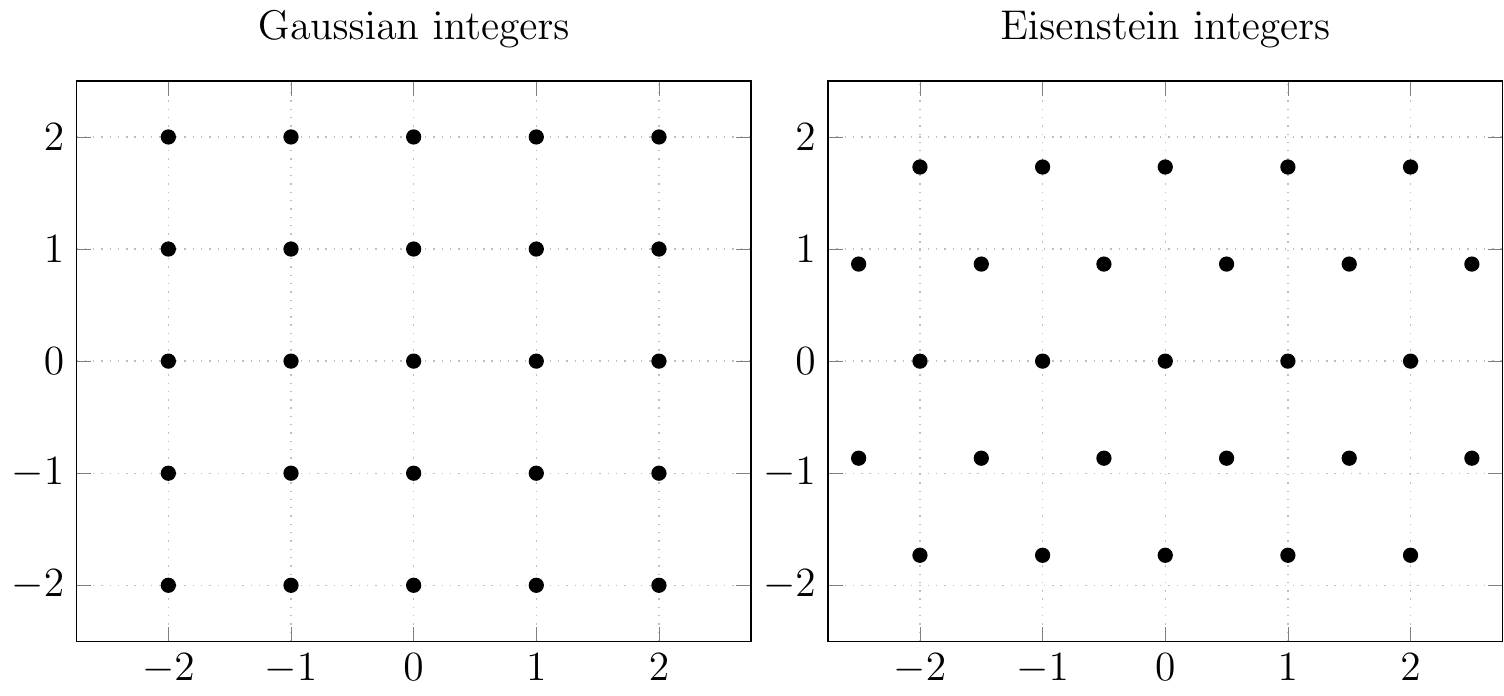}
	\caption{Some Gaussian and Eisenstein integers in 
		the complex plane}
\end{figure}

Additively and multiplicatively, the Gaussian and Eisenstein integers
share many of the properties of the usual integers $\mathbb{Z}$. In particular,
they are archetypal examples of rings other than $\mathbb{Z}$
that are not fields and admit a Euclidean division
\cite{Fraleigh-Ch47,NTA-Ch3}. 
This allows finite fields to be easily constructed as their
residue classes just as the prime field $\mathbb{F}_p$
can be constructed as a set of
residue classes of the usual integers $\mathbb{Z}/p\mathbb{Z}$.

In what follows, we will make claims and provide definitions involving
both the Gaussian and Eisenstein integers. Some will apply to
both and some will be specific to one or the other.
We will use $\Lambda$ to refer to any one of
$\mathbb{G}$ or $\mathbb{E}$ when the definition or result applies to both.

We define a \textit{quantization function} 
$Q_\Lambda \colon \mathbb{C} \longrightarrow
\Lambda$ as in \cite{LNC-Eis} by
\begin{equation}\label{Quantization-Function}
Q_\Lambda(z) = \argmin_{\lambda \in \Lambda} \abs{z-\lambda}
\end{equation}
for all $z\in \mathbb{C}$. We will comment later on
the well-definedness of this function.
Importantly, the quantization can be performed efficiently
for the sets under consideration. For $\Lambda = \mathbb{G}$, we 
have, for all $z\in\mathbb{C}$, \cite{LRA}
\begin{equation*}
Q_\mathbb{G}(z) = 
\left[\Re{z}\right] + \imath\left[\Im{z}\right]
\end{equation*}
where $\left[\cdot\right]$ denotes the usual operation of 
rounding to the nearest integer. For $\Lambda = \mathbb{E}$,
we have the following algorithm from \cite{LNC-Eis}: For any $z \in \mathbb{C}$,
compute 
\begin{align*}
\lambda_1 &= \left[\Re{z}\right] + \sqrt{-3}\left[\frac{\Im{z}}{\sqrt{3}}\right]\\
\lambda_2 &= 
\left[\Re{z-\omega}\right] + 
\sqrt{-3}\left[\frac{\Im{z-\omega}}{\sqrt{3}}\right] + \omega
\end{align*}
which are the nearest points to $z$ in two sets which partition 
$\mathbb{E}$. We then have
\begin{equation*}
Q_\mathbb{E}(z) = \argmin_{\lambda \in \{\lambda_1,\lambda_2\}} \abs{z-\lambda}
\text{.}
\end{equation*}

We will be interested in quantization to a 
sub-lattice $\Lambda' \subseteq \Lambda$
obtained as $\Lambda' = \Pi \Lambda$ where $\Pi \in \Lambda \setminus \{0\}$.
In this case, we have 
\begin{equation*}
Q_{\Lambda'}(z) = 
Q_{\Pi\Lambda}(z) = \Pi \cdot Q_\Lambda\mleft(\frac{z}{\Pi}\mright)
\end{equation*}
for all $z\in\mathbb{C}$. Finally, given $\Lambda$ and
$\Pi \in \Lambda \setminus \{0\}$, we define a
\textit{modulo function} 
$\modulo_{\Pi\Lambda}\colon \Lambda \longrightarrow \Lambda$
as in \cite{Sven}
by
\begin{align*}
\modulo_{\Pi\Lambda}(z) 
&= z - Q_{\Pi\Lambda}(z)\\
&= z -  \Pi \cdot Q_\Lambda\mleft(\frac{z}{\Pi}\mright)
\end{align*} 
for all $z \in \Lambda$. The modulo function
provides a unique representative of the coset in 
the quotient ring $\Lambda/\Pi\Lambda$
to which any $z\in \Lambda$ belongs.
Consider the constellation $\mathcal{A}_{\Pi\Lambda}$ obtained
as 
\begin{equation*}
\mathcal{A}_{\Pi\Lambda}
= 
\{\modulo_{\Pi\Lambda}(z) \mid z \in \Lambda \}\text{.}
\end{equation*}
It can be shown (see, e.g., \cite{Sven,LRA})
that
\begin{equation*}
	\abs{\mathcal{A}_{\Pi\Lambda}} = \abs{\Lambda/\Pi\Lambda} = \abs{\Pi}^2\text{.}
\end{equation*}
Moreover, by the fact
that if $\Pi = rs$ for 
some $r,s\in \Lambda$, then 
$\abs{\Pi}^2 = \abs{r}^2\abs{s}^2$, 
we have the following:
\begin{proposition*}
	If $\abs{\Pi}^2$ is prime in $\mathbb{Z}$,
	then $\Pi$ is prime in $\Lambda$.
\end{proposition*}

If we take $\Pi$ to be a prime in $\Lambda$, then $\Lambda/\Pi\Lambda$
will be isomorphic to $\mathbb{F}_{\abs{\Pi}^2}$ and the same
will be true of the constellation $\mathcal{A}_{\Pi\Lambda}$ under the appropriate 
arithmetic. 
Note that the converse of the previous proposition is not true;
primes in $\Lambda$ corresponding
to non-prime squared absolute value can exist and will allow for 
isomorphisms to finite fields of prime power sizes
(refer to, e.g., \cite{Arbitrary,Sven,LNC-Eis}).
However, we will limit ourselves to prime fields for simplicity.
Importantly, constellations corresponding to prime $\Pi$
give rise to rank-metric-preserving maps.
We will proceed to give some useful
properties of these constellations.

\begin{proposition}\label{Set-Membership-Remark}
	If $z\in\Lambda$, then $z\in\mathcal{A}_{\Pi\Lambda}$
	if and only if $\modulo_{\Pi\Lambda}(z) = z$. Equivalently, 
	if $z\in\Lambda$,
	then $z\in\mathcal{A}_{\Pi\Lambda}$ if and only if
	$Q_{\Lambda}\mleft(z/\Pi\mright) = 0$.
\end{proposition}

We now briefly digress to address a technicality.
In order for Proposition \ref{Set-Membership-Remark}
to be true, we must either have there be no lattice
points
on the boundary of the Voronoi region of zero 
for the quantizer $Q_{\Pi\Lambda}$ or we should 
choose tie-breaking rules for \eqref{Quantization-Function}
carefully as to guarantee that the modulo function
yields unique coset representatives. In this paper,
we will not be considering constellations where this 
is an issue.
We refer the reader to
\cite{NTRU-over-the-Eisenstein-Integers}
for a treatment of this fine point since it
is otherwise important.

Observe that for any $\Pi = a + b\imath\in \mathbb{G}\setminus\{0\}$, we have
$\abs{\Pi}^2 = \abs{a+b\imath}^2 = a^2 + b^2$, and for
any $\Pi = a + b\omega\in \mathbb{E}\setminus\{0\}$, we have
$\abs{\Pi}^2 = \abs{a+b\omega}^2 = a^2 - ab + b^2$. In light 
of this, the following classical results of
number theory will characterize 
the constellation sizes that we can obtain if we insist that
$\abs{\Pi}^2$ is prime in $\mathbb{Z}$. 
The first can be found in \cite{Fraleigh-Ch47}
and the second can be found in \cite{MSE1}.

\begin{theorem*}
	Let $p$ be a prime in $\mathbb{Z}$ greater than $2$. Then, $p = a^2 + b^2$
	for some $a,b\in\mathbb{Z}$ if and only if $p = 4n+1$ for
	some $n \in \mathbb{Z}$.
\end{theorem*}

\begin{theorem*}
	Let $p$ be a prime in $\mathbb{Z}$ greater than  $3$. Then, $p = a^2 - ab + b^2$
	for some $a,b\in\mathbb{Z}$ if and only if $p = 3n+1$ for
	some $n \in \mathbb{Z}$.
\end{theorem*}

Thus, we are restricted to constellations of prime size $p$
where $p$ is $4n+1$ for Gaussian integer constellations or 
or $3n+1$ for Eisenstein integer constellations 
for some $n\in \mathbb{Z}$.
When $p = 12n + 1$ for some $n\in\mathbb{Z}$,
we can choose between a Gaussian and Eisenstein integer constellation.
The Eisenstein integers are more densely packed making them preferable absent any other considerations.

\begin{table}
\centering
	\begin{minipage}{0.48\columnwidth}
		\centering
                \setlength{\tabcolsep}{4pt}
		\begin{tabular}{c | c | c | c}
			$\abs{\Pi}^2$ & $\Pi$ &
			$\abs{\Pi}^2$ & $\Pi$ \\ 
			\hline
			$5$  & $2 + 1\imath$ & $97$ & $9 + 4\imath$\\
			$13$ & $3 + 2\imath$ & $101$ & $10 + 1\imath$\\
			$17$ & $4 + 1\imath$ & $109$ & $10 + 3\imath$\\
			$29$ & $5 + 2\imath$ & $113$ & $8 + 7\imath$\\
			$41$ & $5 + 4\imath$ & $157$ & $6 + 11\imath$\\
			$53$ & $7 + 2\imath$ & $241$ & $4 + 15\imath$\\
			$61$ & $6 + 5\imath$ & $257$ & $1 + 16\imath$\\
			$73$ & $8 + 3\imath$ & $373$ & $7 + 18\imath$\\
			$89$ & $8 + 5\imath$ & $389$ & $10 + 17\imath$\\
                        \multicolumn{4}{c}{ }
		\end{tabular}
		\caption{Some Gaussian primes}
	\end{minipage}\hfill
	\begin{minipage}{0.48\columnwidth}
		\centering
                \setlength{\tabcolsep}{4pt}
		\begin{tabular}{c | c | c | c }
			$\abs{\Pi}^2$ & $\Pi$ &
			$\abs{\Pi}^2$ & $\Pi$\\ 
			\hline
			$7$ & $3 + 1\omega$ & $73$ & $9 + 1\omega$\\
			$13$ & $4 + 1\omega$ & $79$ & $10 + 3\omega$\\
			$19$ & $5 + 2\omega$ & $97$ & $11 + 3\omega$\\
			$29$ & $5 + 2\omega$ & $103$ & $11 + 2\omega$\\
			$31$ & $6 + 1\omega$ & $109$ & $12 + 5\omega$\\
			$37$ & $7 + 3\omega$ & $127$ & $13 + 6\omega$\\
			$43$ & $7 + 1\omega$ & $241$ & $15 + 16\omega$\\
			$61$ & $9 + 4\omega$ & $271$ & $9 + 19\omega$\\
			$67$ & $9 + 2\omega$ & $277$ & $12 + 19\omega$\\
                        \multicolumn{4}{c}{ }
		\end{tabular}
		\caption{Some Eisenstein primes}
	\end{minipage}
\end{table}

It can be shown (see, \cite{Sven,Arbitrary}) that
$\{0,1,\dots,\abs{\Pi}^2-1\}$ are a complete set of
coset representatives so that we have
\begin{equation*}
\mathcal{A}_{\Pi\Lambda}
=\{\modulo_{\Pi\Lambda}(z) \mid z \in \{0,1,\dots,\abs{\Pi}^2-1\}\}\text{.}
\end{equation*}

Finally, we have the following theorem:

\begin{theorem}[Part of Robert  Breusch's 
	Extension of Bertrand's  Postulate \cite{Breusch}]\label{Breusch-Theorem}
	There is at 
	least one prime $p$ of the form $p = 4n+1$ for some $n\in\mathbb{Z}$
	and at least one prime $p'$ of the form $p' = 3n'+1$ for some $n'\in\mathbb{Z}$
	between $m$ and $2m$ for any $m \geq 7$.
\end{theorem}

This means that for any desired constellation size,
we can always find a constellation which is at worst 
twice the size of what is required. Importantly, this
will imply the existence of codes close to the constellation
size lower bounds in Fig.~\ref{Constellation-Size-Bounds}
for arbitrary parameters.

The following theorem
consolidates results from
\cite{ST-Gaussian-Integers,Gabidulin-Space-Time,Gauss-Int-Map-Is-RMP,Arbitrary,Sven}.

\begin{theorem}[Gaussian \cite{ST-Gaussian-Integers,Gabidulin-Space-Time,Gauss-Int-Map-Is-RMP,Arbitrary}
	and Eisenstein \cite{Sven,Arbitrary} Integer Maps Are Rank-Metric-Preserving]
	\label{Gauss-Eis-Int-Map}
	Let $\Lambda$ be  $\mathbb{G}$ or 
	$\mathbb{E}$.
	Let $\Pi$ be a prime in $\Lambda$ with $\abs{\Pi}^2$ a prime
	in $\mathbb{Z}$. 
	Let $\mathbb{F}_{\abs{\Pi}^2} = \mathbb{Z}_{\abs{\Pi}^2} = \{0,1,\dots
	,\abs{\Pi}^2-1\}$. The map $\phi\colon \mathbb{F}_{\abs{\Pi}^2} 
	\longrightarrow \mathcal{A}_{\Pi\Lambda}$
	defined by $\phi(z) = \modulo_{\Pi\Lambda}(z)$ 
	for all $z\in\mathbb{F}_{\abs{\Pi}^2}$
	and evaluated in $\mathbb{C}$ is rank-metric-preserving.
\end{theorem}

A detailed proof of this theorem would be essentially identical
to the proof of a very similar theorem for $\Lambda = \mathbb{E}$
occurring in \cite{Sven} so we will omit it.

Finally, we consider another kind of rank-metric
preserving map which occurs as a special
case of the code construction technique in \cite{Lu-Kumar}.
The proof of the main result in \cite{Lu-Kumar}
can be seen to encompass a proof of the following theorem:

\begin{theorem}[$p$-PSK Map Is Rank-Metric-Preserving \cite{Lu-Kumar}]\label{p-PSK Map}
	\hphantom{a}\\Let $\mathbb{F}_p = \mathbb{Z}_{p} = \{0,1,\dots,p-1\}$
	with $p$ a prime. The
	map $\phi \colon z \mapsto \exp\mleft(\frac{\imath 2\pi z}{p}\mright)$ 
	evaluated in $\mathbb{C}$ is rank-metric-preserving.
\end{theorem}

\subsection{Rate--Diversity Optimal Multiblock Space--Time Codes}
\label{Construction-Sec}

We now  provide explicit constructions
of rate--diversity optimal multiblock space--time codes
which arise as a straightforward consequence
of what has been discussed
thus far. For the cases of $T \geq n_t$ and $T \leq n_t$,
we name the constructions 
\textit{Sum-Rank A (SRA)} and \textit{Sum-Rank B (SRB)}
respectively. In the case of $T = n_t$,
the constructions differ by a transposition of
the sub-codewords and SRB will be preferable for
convenience in the decoding procedure.

\begin{proposition}[SRA Construction]\label{SRA-Construction}
	Fix positive integers $n_t$, $T$, $L$, $d$, and $q$ with 
	$T\geq n_t$, $d\leq Ln_t$, and $q$ a prime power satisfying $q > L$. 
	Let $\mathcal{A}$ be a constellation of size $q$
	and let $\phi\colon \mathbb{F}_{q} \longrightarrow \mathcal{A}$ 
	be a rank-metric-preserving map with
	$\tilde\phi\colon\mathbb{F}_{q}^{n_t\times T} 
	\longrightarrow \mathcal{A}^{n_t\times T}$ the corresponding 
	entrywise map.
	Let $\mathcal{B} = (\beta_1,\dots, \beta_T)$ 
	be an ordered basis of $\mathbb{F}_{q^T}$ over $\mathbb{F}_q$
	and ${M_\mathcal{B}\colon\mathbb{F}_{q^T}^{n_t} 
		\longrightarrow \mathbb{F}_{q}^{T\times n_t}}$
	be a matrix representation map.
	Let $G_1,\dots,G_L \in \mathbb{F}_{q^T}^{(Ln_t-d+1) 
			\times n_t}$ be the sub-codeword 
	generators of an ${[Ln_t,Ln_t-d+1]_{q^T}}$ 
	linearized Reed--Solomon code.
	Then, the $L$-block $n_t \times T$ space--time code
	$\mathcal{X}_\mathsf{SRA}$
	completely over $\mathcal{A}$ defined by
	\begin{multline*}
	\mathcal{X}_\mathsf{SRA} =\\
	\left\{\begin{bmatrix}
	\tilde\phi(M_\mathcal{B}(\mathbf{u}G_1)^{\intercal}) 
	& \cdots 
	& \tilde\phi(M_\mathcal{B}(\mathbf{u}G_L)^{\intercal})
	\end{bmatrix}
	\mmiddle|\mathbf u \in \mathbb{F}_{q^T}^{Ln_t-d+1}\right\}
	\end{multline*}
	has transmit diversity gain $d$ and is rate--diversity optimal.
\end{proposition}
\begin{proof}
	We have ${\abs*{\mathcal{X}_\mathsf{SRA}} = q^{T(Ln_t-d+1)}}$ yielding rate 
	$R = n_t - (d-1)\cdot L^{-1}$ so that $(R,d)$ is an optimal rate--diversity 
	pair. We then have 
	\begin{align*}
	d &= Ln_t - LR + 1\\
	& \geq
	\min_{\substack{\mathbf{u},\mathbf{v} \in \mathbb{F}_{q^T}^{Ln_t-d+1}
			\\\mathbf{u}\neq\mathbf{v}}}
	\sum_{\ell=1}^L
	\rank(\tilde\phi(M_\mathcal{B}(\mathbf{u}G_\ell)^\intercal)-
	\tilde\phi(M_\mathcal{B}(\mathbf{v}G_\ell)^\intercal))\\
	&\geq 
	\min_{\substack{\mathbf{u},\mathbf{v} \in \mathbb{F}_{q^T}^{Ln_t-d+1}
			\\\mathbf{u}\neq\mathbf{v}}}
	\sum_{\ell=1}^L\text{rank}(M_\mathcal{B}(\mathbf{u}G_\ell)^\intercal-
	M_\mathcal{B}(\mathbf{v}G_\ell)^\intercal)\\ 
	&= \min_{\substack{\mathbf{u},\mathbf{v} \in \mathbb{F}_{q^T}^{Ln_t-d+1}
			\\\mathbf{u}\neq\mathbf{v}}}
	\sum_{\ell=1}^L
	\text{rank}(M_\mathcal{B}(\mathbf{u}G_\ell)-
	M_\mathcal{B}(\mathbf{v}G_\ell)) = d
	\end{align*}
	where the first inequality is by the rate--diversity
	tradeoff (Theorem \ref{Rate-Diversity-Tradeoff-Theorem})
	and the last equality is by the MSRD property 
	of linearized Reed--Solomon codes (Theorem \ref{MSRD-Property}).
	Therefore, $\mathcal{X}_\mathsf{SRA}$ achieves the
	rate--diversity pair $(R,d)$.
\end{proof}

We can similarly show the following:

\begin{proposition}[SRB Construction]\label{SRB-Construction}
	Fix positive integers $n_t$, $T$, $L$, $d$, and $q$ with 
	$T\leq n_t$, $d\leq LT$, and $q$ a prime power satisfying $q > L$. 
	Let $\mathcal{A}$ be a constellation of size $q$
	and let $\phi\colon \mathbb{F}_{q} \longrightarrow \mathcal{A}$ 
	be a rank-metric-preserving map with
	$\tilde\phi\colon\mathbb{F}_{q}^{n_t\times T} 
	\longrightarrow \mathcal{A}^{n_t\times T}$ the corresponding 
	entrywise map.
	Let $\mathcal{B} = (\beta_1,\dots, \beta_{n_t})$ 
	be an ordered basis of $\mathbb{F}_{q^{n_t}}$ over $\mathbb{F}_q$
	and ${M_\mathcal{B}\colon\mathbb{F}_{q^{n_t}}^{T} 
		\longrightarrow \mathbb{F}_{q}^{n_t\times T}}$
	be a matrix representation map.
	Let $G_1,\dots,G_L \in \mathbb{F}_{q^{n_t}}^{(LT-d+1) 
		\times T}$ be the sub-codeword 
	generators of an ${[LT,LT-d+1]_{q^{n_t}}}$ 
	linearized Reed--Solomon code.
	Then, the $L$-block $n_t \times T$ space--time code
	$\mathcal{X}_\mathsf{SRB}$
	completely over $\mathcal{A}$ defined by
	\begin{multline*}
	\mathcal{X}_\mathsf{SRB} =\\
	\left\{\begin{bmatrix}
	\tilde\phi(M_\mathcal{B}(\mathbf{u}G_1)) 
	& \cdots 
	& \tilde\phi(M_\mathcal{B}(\mathbf{u}G_L))
	\end{bmatrix}
	\mmiddle|\mathbf u \in \mathbb{F}_{q^{n_t}}^{LT-d+1}\right\}
	\end{multline*}
	has transmit diversity gain $d$ and is rate--diversity optimal.
\end{proposition}

A consequence of the underlying codes
in the SRA and SRB constructions being MDS
is the following:

\begin{corollary*}
	Let $\mathcal{X}$
	be an $L$-block $n_t \times T$ 
	SRA or SRB code completely over $\mathcal{A}$
	with $X$ 
	sampled uniformly at random from
	$\mathcal{X}$. Then, $(X)_{ij}$ is
	uniformly distributed over $\mathcal{A}$
	for $i = 1,\dots,n_t,\,j = 1,\dots, LT$
	and we have
	\begin{equation*}
	\Expect\mleft[\lVert X \rVert_\mathsf{F}^2\mright] 
	= 
	\frac{n_tLT}{\abs{\mathcal{A}}}
	\sum_{a \in \mathcal{A}}\abs{a}^2\text{.}
	\end{equation*}
\end{corollary*}

We can also interpret these constructions in 
terms of the signalling complexity perspective 
provided in Section \ref{sig-comp-subsection}.
In particular,
we have the following:

\begin{corollary*}
	Fix positive integers $n_t$, $T$, and $L$.
	Fix some $R_\mathsf{b/tx}$ and
	some $\varepsilon$ satisfying $0 \leq \varepsilon < 1$.
	Let 
	\begin{equation}\label{The-B}
	B = 
	\max\{7,\,L+1,\,\ceil{F_\varepsilon(L\cdot\min\{n_t,T\})}\}
	\end{equation}
	where $F_\varepsilon$ is as defined in \eqref{Constel-Size-LB-Function}.
	Then, there exists a constellation $\mathcal{A} \subset \mathbb{G}$
	(or $\mathcal{A} \subset \mathbb{E}$)
	such that 
	$B \leq \abs{\mathcal{A}} < 2B$
	and an $L$-block $n_t\times T$ space--time 
	code completely
	over $\mathcal{A}$ 
	achieving a transmit diversity gain of
	$d = \floor*{(1-\varepsilon)\cdot L\cdot\min\{n_t,T\}}$
	and a bpcu/tx rate of at least $R_\mathsf{b/tx}$.
\end{corollary*}

As will be seen later in Fig.~
\ref{Lab-Constellation-Size-Bounds},
it will usually be the case for the parameters
of interest that the $L+1$ in \eqref{The-B}
does not come into play. Consequently,
we will be able to construct codes near the 
constellation
size lower bounds in Fig.~\ref{Constellation-Size-Bounds}.

We conclude this section by commenting
briefly on  connections to other constructions.
In the case of $T\geq Ln_t$, we can replace
linearized Reed--Solomon codes with
Gabidulin codes and, by Corollary
\ref{Finer-Partitions-Preserve-Total-Rank-Finite-Field-Corollary}, 
we will have a rate--diversity optimal code
which coincides with a special case of 
the multiblock construction in \cite{Lu-Kumar}
when a $p$-PSK map (Theorem \ref{p-PSK Map}) is used. 
Moreover, 
in the case of $L = 1$, linearized Reed--Solomon 
codes become Gabidulin codes and these codes
will become identical to those in
\cite{ST-Gaussian-Integers,Gabidulin-Space-Time} when 
a Gaussian integer constellation is used 
and identical to those in \cite{Sven} when an
Eisenstein integer constellation is used.

\section{Simulations}\label{Simulations-Sec}

In this section, we study the 
performance of the proposed codes under ML
decoding in comparison to competing codes. 
We defer the 
description of the decoders 
which were used to obtain the simulation results
to Section \ref{Decoding-Sec}. 

For the entirety 
of this section, we will have $n_r = n_t = T$.
We will refer to Gaussian and Eisenstein integer 
constellations of size $p$ as $p$-Gauss.\ and $p$-Eis.\
constellations respectively.

\subsection{The Case of $L=2$}\label{2-Block-Sim-Sec}

We begin with the case of coding across two fading blocks with $n_t = 2$
and consider
a bpcu/tx rate of $2$. 
Since the proposed codes require certain prime constellation sizes,
they will not exactly achieve this bpcu/tx rate. We will typically
choose constellations so that the bpcu rate is at least what is required.
Moreover, the total available transmit diversity gain in this case is $Ln_t = 4$.

A variety of codes are considered and the codeword error rate (CER)
versus SNR curves are provided in Fig.~\ref{2-Block-CERs-Small-Constellations}.
We also provide all of the corresponding constellations in Figures \ref{CDA-Constellations}
and \ref{My-Small-Constellations}. We will proceed to describe these codes.

\begin{figure}[t]
	\centering
	\includegraphics[width=\columnwidth]{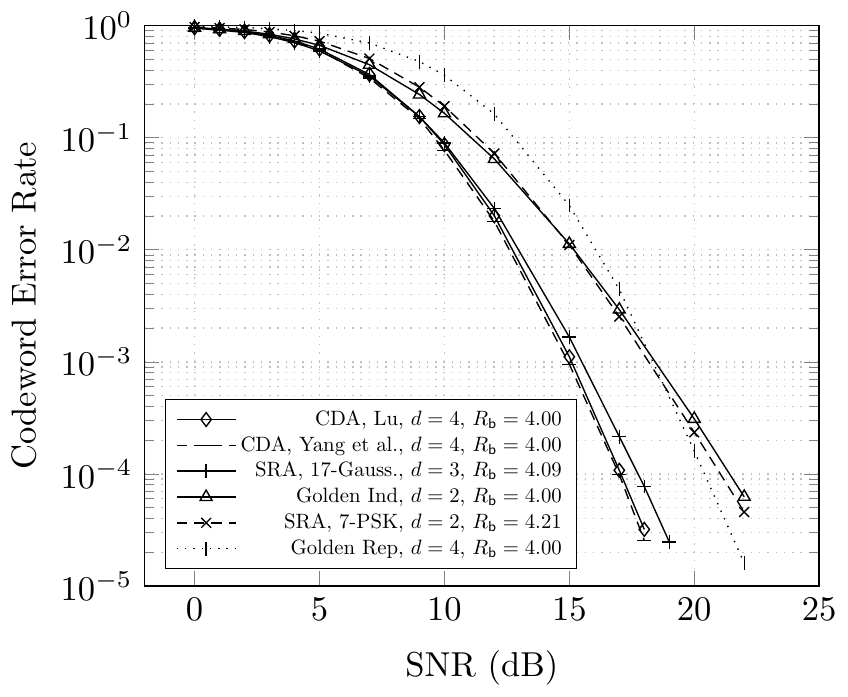}
	\caption{Performance of coding across $L = 2$ fading blocks of a $2 \times 2$ MIMO channel}\label{2-Block-CERs-Small-Constellations}
\end{figure}

We provide two $d = 4$ CDA-based $2$-block linear dispersion codes 
as a reference both using $4$-QAM input constellations to 
achieve a bpcu rate of $4$. 
One is due to Lu \cite{Hsiao-Code} and 
the other is due to Yang and Belfiore \cite{Sheng}.
We then provide our SRA construction with $d = 3$
and a $17$-Gauss.\ constellation
to achieve a bpcu rate of $4.09$. 
It can be seen that the SRA code
performs worse but the gap is less than $1$ dB
at a CER of $10^{-4}$. In exchange, the constellation
size is significantly smaller as can be seen in 
Figures \ref{CDA-Constellations} and \ref{My-Small-Constellations}.
Moreover, we see from \eqref{Rate-Diversity-Tradeoff-Full-Diversity} 
that the code due to Yang and Belfiore \cite{Sheng} is
also rate--diversity optimal hence the smaller constellation
than that of Lu \cite{Hsiao-Code}. 

We further provide as a reference the performance of a 
\textit{single-block} full diversity linear dispersion code, namely
the Golden Code \cite{Golden-Code}. It is used in our two-block 
channel in two ways. The first is with sending
independent codewords in each
fading block with a $4$-QAM input constellation 
to achieve a bpcu rate of $4$ and $d = 2$.
We refer to this code, which is 
the Cartesian product of two Golden Codes, as \textit{Golden Ind.}
This yields the performance of coding across a single fading block
as the code was designed for.
We further provide the result of repeating a Golden
Code codeword 
in each fading block. This is referred to as 
\textit{Golden Rep.}
This represents a trivial way of obtaining
a full diversity, i.e., $d = 4$ code. 
To compensate for the rate loss, we use a 
$16$-QAM input constellation to obtain a bpcu
rate of $4$. As can be seen in Fig.~\ref{2-Block-CERs-Small-Constellations},
this code performs quite poorly which verifies, 
in some sense,
the nontriviality
of the other $d = 4$ and $d = 3$ codes
intended for multiblock channels.

Note that the Golden Ind and Golden Rep codes
are \textit{multiblock} codes constructed from the
\textit{single-block} Golden Code and can be analyzed 
as such. In particular, one can verify that the Golden Rep code
is rate--diversity optimal. On the other hand, the Golden Ind
code is not rate--diversity optimal. It has a rate of
$1$ while the maximum possible for $d = 2$ is $1.5$. Again,
this is in an $L = 2$ setting and would not be
true if we were analyzing the Golden Code 
itself as 
a single-block code.

In light of this, we consider a rate--diversity
optimal $d = 2$ code, particularly
our SRA code with a $7$-PSK
constellation achieving a bpcu rate of $4.21$. 
This performs comparably to the Golden Ind code.
The higher bpcu rate achieved with a $7$-PSK constellation
is due to its rate--diversity optimality.
Alternatively, we can interpret it as a result of exploiting
the possibilities admitted by the sum-rank metric. 
For example,
rather than always needing rank-distance-2 sub-codewords in
at least one fading block which is the case for the Golden Ind code,
rank-distance-1 sub-codewords may
be transmitted in the first fading block provided that 
rank-distance-1 \textit{or} rank-distance-2 sub-codewords occur in the next
fading block. Indeed, such combinations do occur in the codebook.

\begin{figure}[t]
	\centering
	\includegraphics[width=\columnwidth]{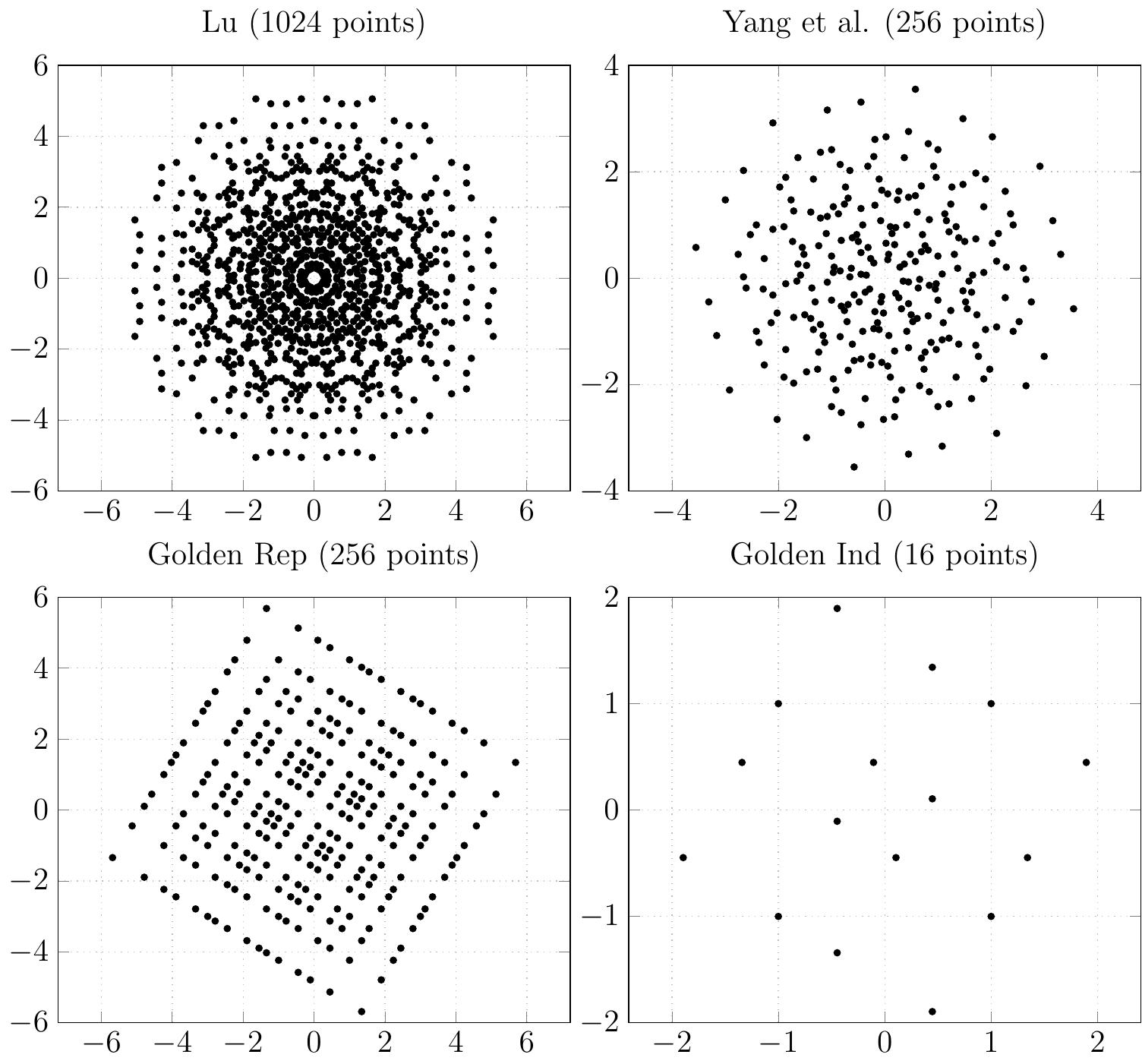}
	\caption{Constellations for CDA-based codes}\label{CDA-Constellations}
	\vspace{0.25cm}
	\includegraphics[width=\columnwidth]{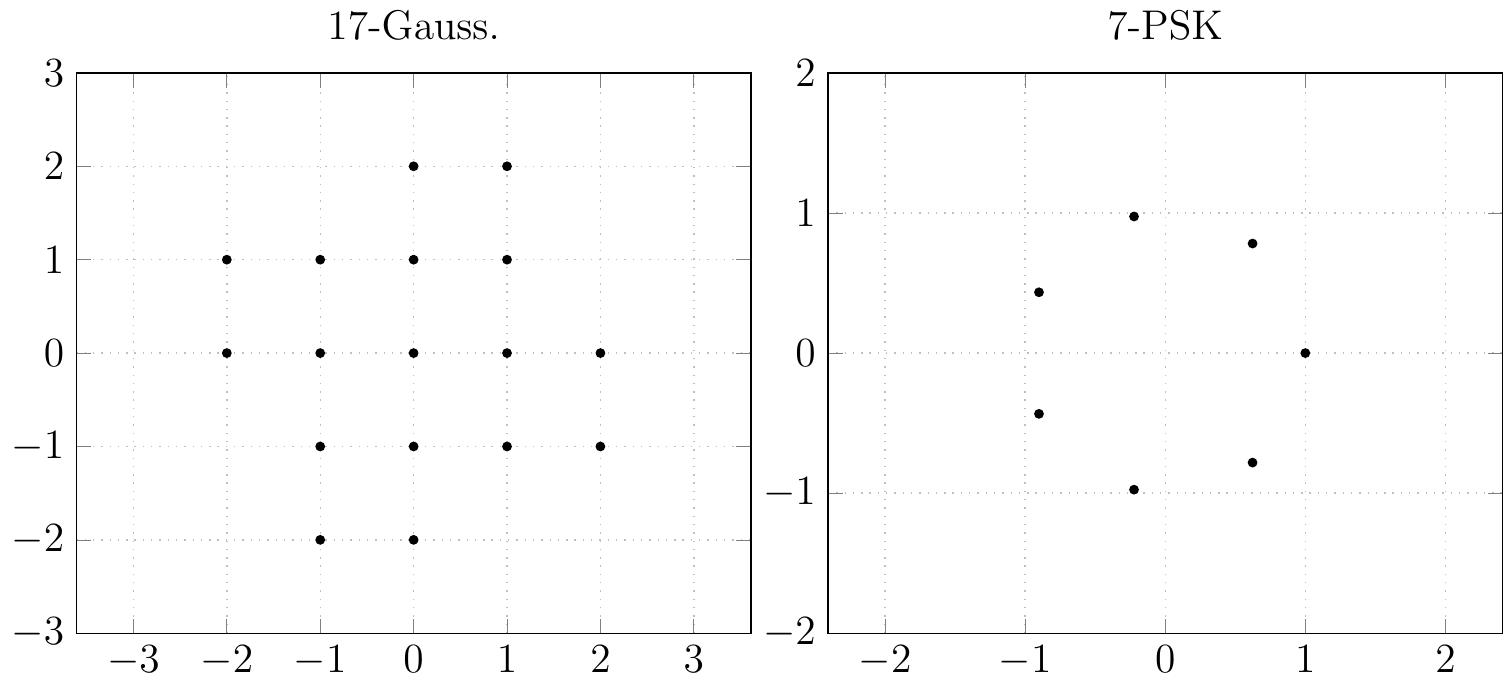}
	\caption{Constellations for proposed codes}\label{My-Small-Constellations}
\end{figure}

Next, we consider increasing the
bpcu/tx rate to $4$
which corresponds to a bpcu rate of $8$. 
This is achieved by using $16$-QAM input
constellations for the CDA-based codes
of Lu \cite{Hsiao-Code} and Yang and Belfiore \cite{Sheng}.
We compare these to our SRB codes with $d = 3$ and a few different constellations.
The results and constellations for our codes are provided in
Figures \ref{2-Block-CERs-Big-Constellations} and \ref{My-Big-Constellations}.
The constellations for the competing codes cannot be 
easily plotted or found in this case, but the constellation
size lower bound says that they should have a size of
at least $2^{Ln_tR_\mathsf{b/tx}} = 2^{16} = 65536$ points.
In this case, the performance gap between our codes and the
competing codes becomes worse, but our constellation sizes
are again significantly smaller. Besides this, the 
results are self-explanatory and the better point density
of the Eisenstein integers is evident.

\begin{figure}[t]
	\centering
	\includegraphics[width=\columnwidth]{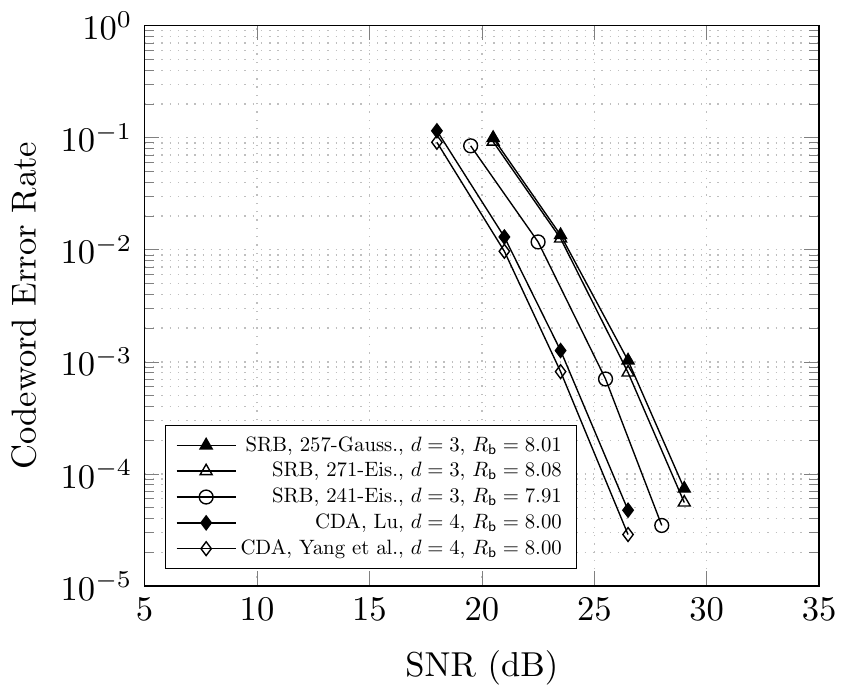}
	\caption{Performance of coding across $L = 2$ fading blocks of a $2 \times 2$ MIMO channel}
	\label{2-Block-CERs-Big-Constellations}
\end{figure}
\begin{figure}[t]
	\centering
	\includegraphics[width=\columnwidth]{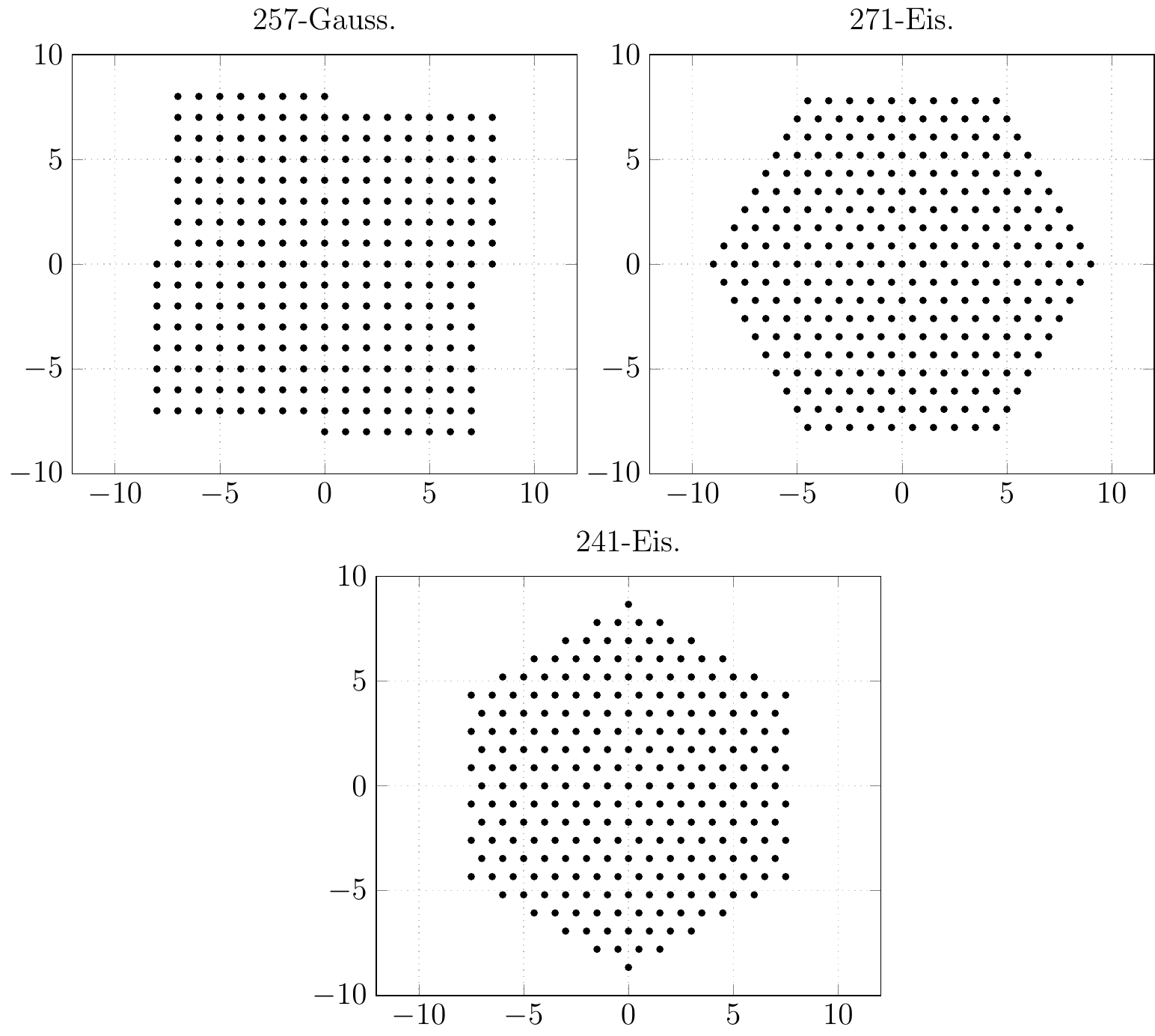}
	\caption{Constellations for proposed codes}\label{My-Big-Constellations}
\end{figure}

\subsection{The Case of $L=4$}

We now consider increasing $L$ to $4$
with $n_t = 2$ and a bpcu/tx rate of $2$.
In this case, the total available transmit
diversity gain is $Ln_t = 8$.

We provide the performance of a
$4$-block $d = 8$ 
CDA-based linear dispersion code
due to Yang and Belfiore \cite{Sheng}.
From our construction, we provide $d = 5$
and $d = 4$ SRB codes. Notably, our
$d = 5$ code outperforms the $d = 8$ code with no sign
of a closing performance gap over the range
of simulated SNRs. Moreover,
the $d = 8$ code has a constellation with at least $65536$
points compared to $17$ for our $d = 5$ code.

In Fig.~\ref{Lab-Constellation-Size-Bounds}, 
we provide constellation size lower
bounds for $n_t = 2$ and $R_\mathsf{b/tx} = 2$
as a function of $L$ and some
of the constellation sizes achieved by
the codes discussed thus far. This also
shows when the $\abs{\mathcal{A}} \geq L + 1$
requirement of our construction might come into play.

\begin{figure}[t]
	\centering
	\includegraphics[width=\columnwidth]{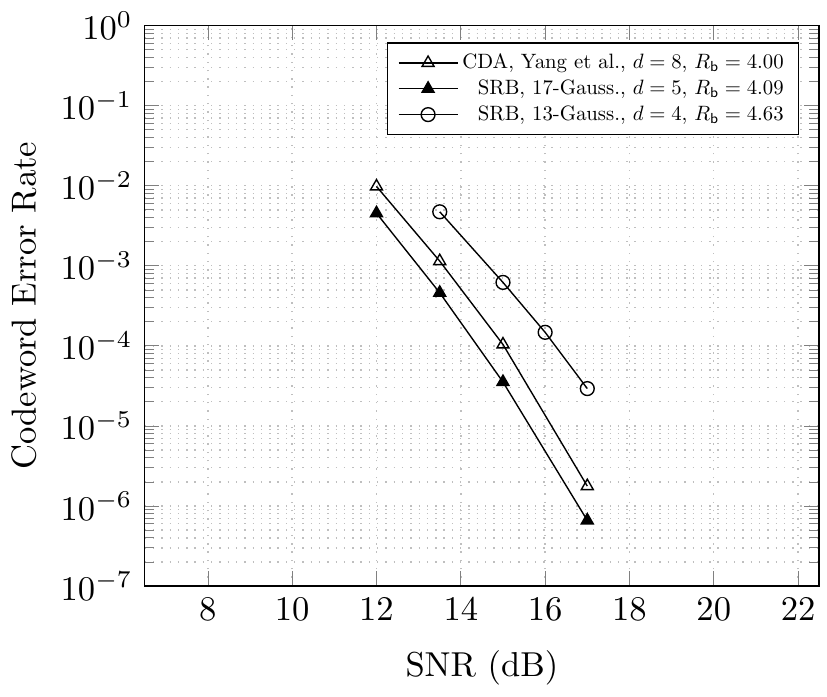}
	\caption{Performance of coding across $L = 4$ fading blocks of a $2 \times 2$ MIMO channel}
\end{figure}

\begin{figure}[t]
	\centering
	\includegraphics[width=\columnwidth]
	{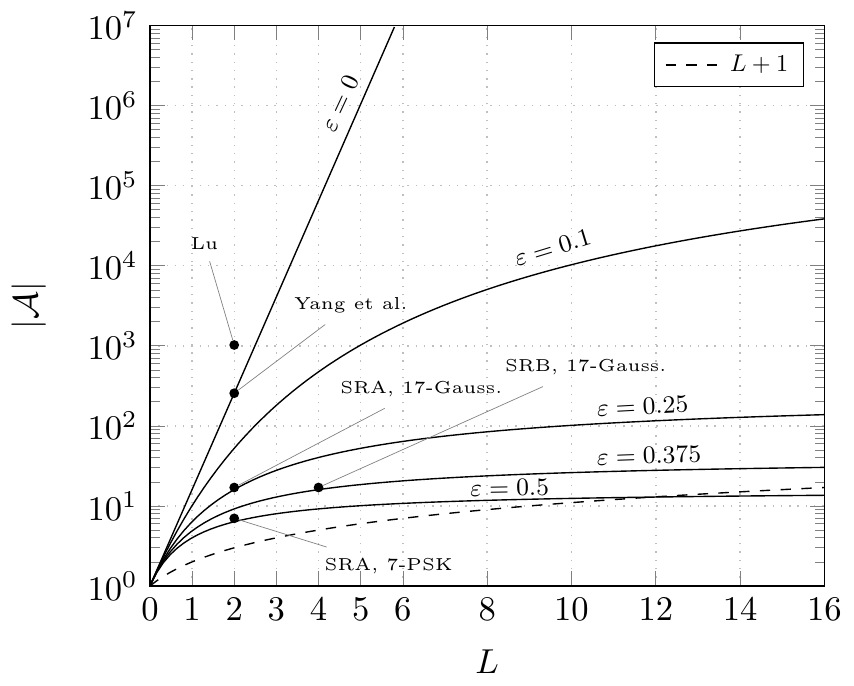}
	\caption{Constellation size lower bounds for 
		$R_\mathsf{b} = 4$, $n_t = 2$, 
		and $d =\ceil*{(1-\varepsilon)Ln_t}$}\label{Lab-Constellation-Size-Bounds}
\end{figure}

\subsection{The Case of $L=1$}\label{1-Block-Sim-Sec}

We now consider the case of $L = 1$.
In this case, our SRA and SRB constructions
recover the known space--time code constructions
in \cite{ST-Gaussian-Integers,Gabidulin-Space-Time,Sven,Arbitrary}
based on Gabidulin codes. However, in these
works, no ML decoding techniques apart from 
exhaustive search are provided
and there is consequently no ability to compare
the performance of the codes to high-rate linear dispersion 
codes.
Facilitated by the decoding techniques which will
be introduced in Section \ref{Decoding-Sec}, 
we are able to decode these codes for the first time
at sufficiently high
bpcu rates to enable such a comparison.

In the case of $n_t = 2$, our codes
have little to offer since the only
options for the diversity gain are $2$
which is full diversity and $1$ which is 
achieved by uncoded
signalling. Therefore, we skip to $n_t = 3$ in which case
we have a total 
available diversity gain of $n_t = 3$. We 
consider target bpcu/tx rates of $2$ and $3$
and compare to the 
Perfect $3\times 3$ codes from 
\cite{Perfect-Codes} using $4$-HEX and $8$-HEX
input constellations which are also defined in \cite{Perfect-Codes}.
The results are provided in Fig.~\ref{3x3}. In this case,
there are no admissible constellation sizes
for our construction that get us close enough to the
desired rate so the comparisons are not very fair.
Nonetheless, one can see that the performances are
close and our codes,
as usual, have a constellation size advantage.

\begin{figure}[t]
	\centering
	\includegraphics[width=\columnwidth]{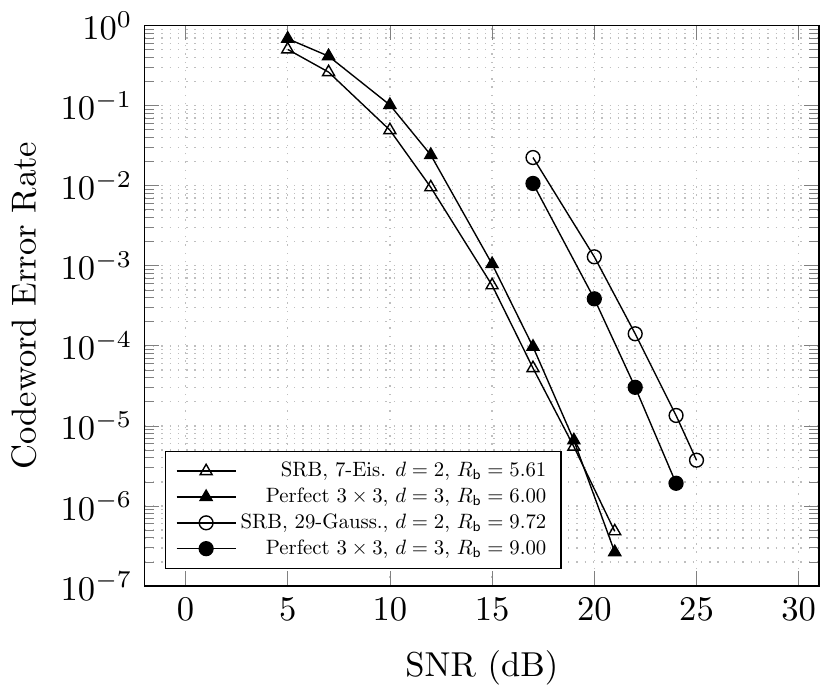}
	\caption{Performance of single-block coding on a $3 \times 3$ MIMO channel}
	\label{3x3}
\end{figure}

We proceed to the case of $n_t = 4$ where we compare to
the Perfect $4\times 4$ code from \cite{Perfect-Codes}.
We further compare to a code which improves upon this one
which we refer to as \textit{Improved $4 \times 4$}. 
This code is from \cite{Improved4} which also credits 
\cite{Improved4-Origin}. 
We consider a target bpcu/tx rate of $2$ which corresponds
to a bpcu rate of $8$ and $4$-QAM input constellations
for the competing codes.
In this case, our $d = 3$ code with a $17$-Gauss.\ constellation
outperforms both codes all the way up to a 
CER of $10^{-6}$ and with a bpcu rate and constellation
size advantage. 
On the other hand, a closing 
performance gap can be seen due to the 
diversity gain difference.
Furthermore, the fact that the CERs for
our code are only provided for high SNRs
foreshadows another issue: this is that
they can have a higher decoding complexity.

\begin{figure}[t]
	\centering
	\includegraphics[width=\columnwidth]{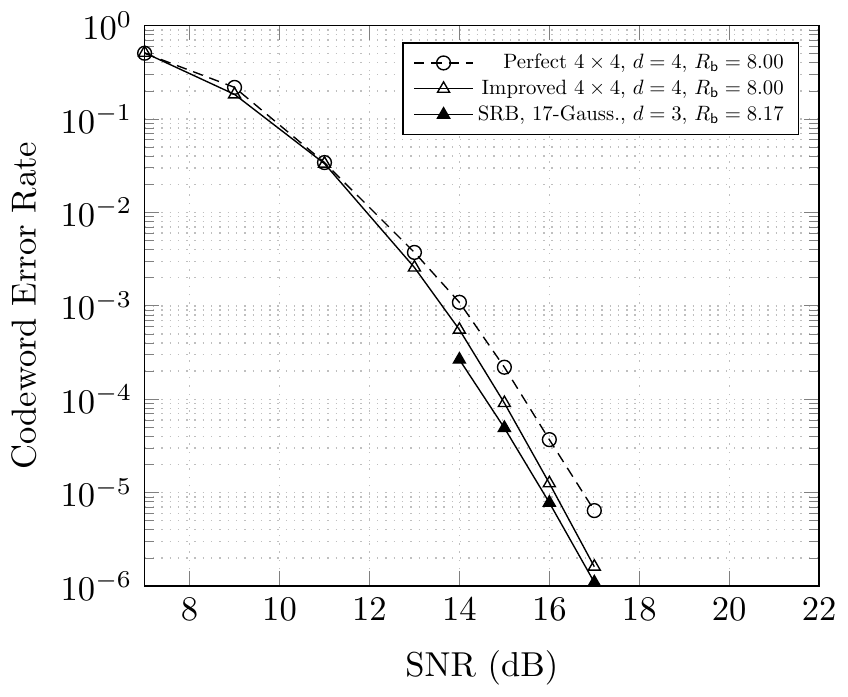}
	\caption{Performance of single-block coding on a $4 \times 4$ MIMO channel}
	\label{4x4}
\end{figure}

\section{Decoding}\label{Decoding-Sec}

In this section, we provide an ML decoder for the proposed
codes which combines ideas from the sequential decoding
of linear block codes, optimal MIMO detection, and sphere
decoding of linear dispersion codes. For concreteness of 
exposition, we describe a single family of
decoders. 
However, in the process, we provide methods for more generally
adapting some of the known techniques for the decoding
of linear dispersion codes to both our codes
and identically or
similarly constructed
space--time codes such as those of 
\cite{Sven,ST-Gaussian-Integers,Gabidulin-Space-Time,Lu-Kumar}.
The decoders to be considered are based on 
tree search
algorithms which
have good average-case complexity
at high SNRs and 
exponential worst-case complexity
in the codeword dimensions, as is
typical.

While we will 
comment extensively on 
how the decoding
problem for the proposed codes 
compares to the decoding 
problem 
for the competing 
linear dispersion codes, 
we will not attempt 
to quantify the difference 
in computational complexity.
Rather, we 
will have the less ambitious goal
of providing a starting point
for decoding and giving a sense 
of the challenges involved 
in comparison to linear dispersion codes.
The central point to be
made regarding this matter is as
follows:
There is one fundamental advantage and one 
fundamental disadvantage in the decoding 
problem for an SRA/SRB code in comparison
to the decoding problem for a comparable 
linear dispersion code.
The fundamental
advantage is that the effective channel
matrix when decoding the SRA/SRB code 
is block diagonal and there
are many ways to exploit this.
The fundamental disadvantage is that
the tree structure for the SRA/SRB
code is inherently worse.
In particular, we must perform the decoding
over a larger alphabet and over a search
space with unfavorable
boundaries. We also note that 
this point as well as the entire
content of this section applies
in the single-block, i.e., $L=1$
setting as well. 

The extent to which
the advantage can compensate for
the disadvantage is an open question
and depends on the code parameters. 
Nonetheless, the advantage can
certainly be exploited enough 
to make the decoding feasible.
Moreover, this question neglects
the potential benefits of smaller
constellations to the implementation 
complexity
so is not necessarily the pertinent 
question from an engineering perspective.

\subsection{Summary of Existing Decoders}

We begin by commenting on existing methods
for decoding the proposed codes.
In \cite{Sven}, the authors consider a 
suboptimal algebraic decoding technique for 
Gabidulin-based space--time codes which is, 
in principle, adaptable to SRA/SRB codes 
via the associated algebraic decoding methods 
for linearized Reed--Solomon codes
\cite{Interpolation,Umberto-Multishot}, 
but the performance is quite far from ML. 
To the best of the authors' knowledge, 
there has been no other attempt at decoding 
space--time codes similar to the proposed codes
aside from that in \cite{Sven} and exhaustive search. 
With that said,
the decoding methods to be proposed generalize
existing methods in fairly straightforward ways.
The methods to be built upon will be referred to 
as they occur in the upcoming development of a decoder.

\subsection{The Vanilla Stack Decoder}

In this subsection, we will describe a 
generic decoding technique and 
apply it to the decoding of the SRB code 
in the special
case of $n_r = n_t$ with
a focus on Gaussian or Eisenstein
integer constellations. We will 
stick to this special case to simplify
notation and exposition but we will
comment on how deviations from this special
case can be handled later.

We will begin by considering 
a general formulation of a
minimum 
cost path finding problem
on a tree. Our 
decoding problem will turn out 
to be a special case of this problem.
The initial algorithm to be described 
which we refer to as a \textit{vanilla
	stack decoder}
is essentially a MIMO analogue
of the
Zigangirov--Jelinek stack decoding
algorithm \cite{Stack1,Stack2}.
Later on, it will evolve into an analogue
of a 
stack decoder with
variable bias-term (VBT) metric \cite{VBT-Metric}
which itself constitutes an instance of the
A* algorithm \cite{A-star}. 
The fundamental
ideas here are well-known and are constantly
rediscovered in the literature. While the 
aforementioned literature was the starting
point for the development of the proposed decoder, 
the closest
pre-existing decoders to what we will describe
are likely to 
be found in the sections in \cite{Fifty-Year-MIMO}
relating to ML MIMO detection and the approaches described in
\cite{Tree-Search-Decoding,ML-CLPS}. 

Following the introduction of the vanilla 
stack decoder,
we will build upon
it with improvements
that exploit structure
that is more specific to
the decoding problem 
and to the code. These improvements
will be quantified in simulations
in Section \ref{Decoding-Complexity-Simulations}.

We will now consider the problem
of generically decoding a tree code.
In this subsection and the others
that will follow, we will make 
use the of the language
and notation of strings rather than vectors. 
Adjacent symbols (or characters) will represent concatenated strings
rather than multiplication in $\mathbb{C}$ unless
otherwise indicated or suggested by context.

Consider a code $\mathcal{S}$ of length $n$ over
a $q$-ary alphabet $\mathcal{A} = \{a_1,a_2,\dots,a_q\}$. 
That is, some $\mathcal{S}\subseteq \mathcal{A}^n$ with
$\abs{\mathcal{S}} > 1$. Any codeword $s \in \mathcal{S}$
is a string of length $n$ so that $s = s_1s_2\cdots s_n$
where $s_1,s_2,\dots,s_n \in \mathcal{A}$.
Denote by $\mathcal{A}^*$ the 
\textit{Kleene closure} of $\mathcal{A}$. This is the set
of all strings of any length over $\mathcal{A}$. That is, 
\begin{equation*}
\mathcal{A}^*
= \{\epsilon,a_1,\dots,a_q,a_1a_1,\dots,a_1a_q,
a_2a_1,a_2a_2,a_2a_3,\dots\}
\end{equation*}
where $\epsilon$ here denotes the \textit{empty string} which
is the identity element of concatenation. That is, 
$\epsilon a = a\epsilon = a$. The length
of a string will be denoted by $\abs{\cdot}$ and the empty
string has length zero, i.e., $\abs{\epsilon} = 0$.

Denote by $\mathcal{S}^*$ the set of all prefixes of codewords 
in $\mathcal{S}$. That is, 
\begin{equation*}
\mathcal{S}^* 
=
\{p\in \mathcal{A}^* \mid pt \in \mathcal{S}\text{ for some }
t \in \mathcal{A}^*\}\textit{.}
\end{equation*}
Note that $\mathcal{A}^*$ is an infinite set
while $\mathcal{S}^*$ is finite. Observe that
if $p \in \mathcal{S}^*$ and $\abs{p} = n$, 
then $p \in \mathcal{S}$.

Denote by $\wp(\mathcal{A})$ the power set of $\mathcal{A}$. This
is the set of all subsets of $\mathcal{A}$.
Define the function $E \colon \mathcal{S}^* \longrightarrow \wp(\mathcal{A})$
by 
\begin{equation}\label{E-Function}
E(p) = \{a\in\mathcal{A}\mid pa \in \mathcal{S}^*\}\text{.}
\end{equation}
Thus, given a prefix of a codeword, this function
gives us the set of all possible candidates for the next character.
Observe that if $p \in \mathcal{S}^*$ and $pb \in \mathcal{S}$ for some 
$b \in \mathcal{A}$, then $b \in E(p)$.
Moreover, if $p \in \mathcal{S}$, then $E(p) = \emptyset$.

Assuming that we can easily compute this function for
our code, i.e., that we can easily enumerate 
the elements of this set, we can easily construct
a representation of our code as a $q$-ary tree of 
depth $n$. Each node represents a prefix of a 
codeword and the leaf nodes correspond to the codewords.
In particular, we can label the root node by $\epsilon$.
Its children at depth $1$ are then labelled by the elements 
of $E(\epsilon) = \{a \in \mathcal{A}\mid a \in 
\mathcal{S}^*\}$. Given a child $c$ at depth $1$,
we can find its children at depth $2$ as $E(c) = \{a\in\mathcal{A}\mid ca \in \mathcal{S}^*\}$, 
and so on.

Consider 
$n$ functions
$f_i\colon \mathcal{A}^i \longrightarrow \mathbb{R}_{\geq 0}$
for $i = 1,2,\dots,n$.
Define the cost function
$C\colon \mathcal{S} \longrightarrow \mathbb{R}_{\geq 0}$ 
by 
\begin{align}
C(a_1a_2\cdots a_n)
&= f_1(a_1) + f_2(a_1a_2)
+ \dots + f_n(a_1a_2\cdots a_n)\nonumber\\
&= \sum_{j = 1}^n f_j(a_1a_2\cdots a_j)\label{Causal-Cost-Function}
\end{align}
for all $a_1a_2\cdots a_n\in\mathcal{S}$ 
(with $a_1,a_2,\dots,a_n\in\mathcal{A}$). We
refer to such a function as a \textit{causal cost function.}
Given such a function, we can replace it by 
one which is extended to arbitrary length prefixes.
Define the \textit{extended causal cost function}
$C\colon \mathcal{S}^*\longrightarrow \mathbb{R}_{\geq 0}$ 
for arguments of length $i$
by 
\begin{align}
C(a_1a_2\cdots a_i)
&= f_1(a_1) + f_2(a_1a_2)
+ \dots + f_i(a_1a_2\cdots a_i)\nonumber\\
&= \sum_{j = 1}^i f_j(a_1a_2\cdots a_j)
\end{align}
for all  $a_1a_2\cdots a_i\in\mathcal{S}^*$  
(with $a_1,a_2,\dots,a_i\in\mathcal{A}$)
for $i = 1,2,\dots,n$.
Observe that the function
can now be defined 
recursively for each argument by
\begin{equation*}
C(a_1a_2\cdots a_i) = 
C(a_1a_2\cdots a_{i-1})  
+ 
f_i(a_1a_2\cdots a_i)
\end{equation*}
and
\begin{equation*}
C(a_1) = f_1(a_1)
\end{equation*}
for $i = 2,3,\dots,n$.

Suppose that we are interested in solving
\begin{equation*}
\min_{s\in\mathcal{S}} C(s)\text{.}
\end{equation*}
A problem of this form can be solved by a best-first 
tree search
algorithm like the Zigangirov--Jelinek stack decoding
algorithm \cite{Stack1,Stack2}. Traditionally, 
these algorithms are used in the context
of binary codes ($\abs{\mathcal{A}} = 2$) and where the cost
function is simple. In particular, we usually have
no dependence on previous symbols in the additive
decomposition of the cost function. That is, 
$f_i(a_1a_2\cdots a_i) = g_i(a_i)$ 
for some $g_i\colon \mathcal{A} \longrightarrow
\mathbb{R}_{\geq 0}$.
Fortunately, the generalization to allow
for arbitrary causal cost functions is
effortless.

We will now describe this algorithm. The 
algorithm will make use of a priority queue 
data structure. This is a stack with
\textit{push} and \textit{pop}
operations. However, rather than having a last-in first-out
(LIFO) order, the elements are paired with priority measures
and sorted accordingly so that the element
with the highest priority comes out first.
Each entry in our priority queue will be a codeword
prefix $p \in \mathcal{S}^*$ and a priority
measure which will be its cost $C(p)$. The highest 
priority element will be the one with the smallest cost.
The algorithm can now be described very simply:
\begin{itemize}
	\item Start by pushing every $a \in E(\epsilon)$
	onto the priority queue with its respective
	cost $C(a) = f_1(a)$.
	\item Pop an element from the priority queue 
	which will consist of a prefix $p$ and its priority measure
	$C(p)$.
	\item For every $t \in E(p)$, push $pt$ onto the priority queue
	with its priority measure $C(pt) = C(p) + f_{\abs{p} + 1}(pt)$.
	\item Repeat the previous two steps until 
	a prefix $s$ with $\abs{s} = n$ is popped.
\end{itemize}

For a proof that this algorithm works, i.e., 
terminates and yields the minimum cost codeword,
the reader is referred to \cite{My-Thesis}. Alternatively,
the reader may simply recognize this as an instance
of a standard best-first search algorithm like A* \cite{A-star}.

\subsubsection*{Application to the Decoding of the Proposed Codes}

We now proceed to the application of this
to the decoding of our space--time codes.
We can associate an
$n_t \times LT$ codeword matrix
with a length-$LTn_t$ 
string via the bijection defined by: 
\begin{equation*}
\begin{bmatrix}
x_1 & x_{n_t + 1} & \dots & x_{(LT-1)n_t + 1} \\ 
x_2 & x_{n_t + 2} & \dots & x_{(LT-1)n_t + 2} \\
\vdots & \vdots  & \ddots & \vdots \\
x_{n_t} & x_{2n_t} & \dots & x_{LTn_t} 
\end{bmatrix}
\mapsto
x_1x_2\cdots x_{LTn_t}\text{.}
\end{equation*}
Our code $\mathcal{X}$ is thus in bijection
with a code $\mathcal{S}$ as defined in the previous
section. We will also use a similar representation
for our receive matrix which will also be $n_t \times LT$
since we are assuming that $n_r = n_t$.

The first order of business is to
find the function $E$. Since the underlying
code is MDS, this is easy. Let $k$ be the dimension
of the linearized Reed--Solomon code ($k = LT - d + 1$ for
the SRB construction) and assume that we are using a
systematic generator matrix. Then, for all $p \in \mathcal{S}^*$
such that $\abs{p} < kn_t$,
we have $E(p) = \mathcal{A}$.
On the other hand, for any 
$p \in \mathcal{S}^*$ with $\abs{p} \geq kn_t$,
the entire remaining symbols are parity symbols
which can be computed from the first $kn_t$ symbols
so we have $\abs{E(p)} = 1$, i.e., there is only 
one possible 
continuation of the prefix. Fig.~\ref{My-Tree} illustrates
the structure of the resulting code tree.

It remains to verify that
the cost function has the appropriate form.
This is merely a matter of performing QL
decompositions. 
Recall our convention that
$X = \begin{bmatrix} X_1 & X_2 & \cdots & X_L \end{bmatrix}$.
For $\ell = 1,2,\dots, L$, 
we have
\begin{equation*}
\lVert Y_\ell - \rho H_\ell X_\ell \rVert_\mathsf{F}^2
=
\lVert \tilde Y_\ell - \mathbf{L}_\ell X_\ell \rVert_\mathsf{F}^2
\end{equation*}
where $\rho H_\ell = Q_\ell \mathbf{L}_\ell$ and 
$\tilde Y_\ell = Q_\ell^\dagger Y_\ell$ 
where $Q_\ell$ is a unitary matrix and 
$\mathbf{L}_\ell$ is a lower-triangular matrix.
Moving forward, we will drop the tilde from the $\tilde Y_\ell$.
Denote by $\ell(j)$ the index of
the fading block corresponding the $j$th column of $X$, i.e.,
\begin{equation*}
\ell(j) = \floor*{\frac{j-1}{T}} + 1
\end{equation*}
for $j = 1,2,\dots, LT$.
For $m = 1,2,\dots,n_tLT$,
define the functions
$g_m \colon \mathcal{A}^{m-\floor*{\frac{m-1}{n_t}}n_t} \longrightarrow 
\mathbb{R}_{\geq 0}$ by
\begin{multline*}
	g_m
	(x_{\floor*{\frac{m-1}{n_t}}n_t+1}
	x_{\floor*{\frac{m-1}{n_t}}n_t+2}\cdots \,x_{m}) 
	= \\
	\abs{y_{m} 
		- 
		\hspace*{-3ex}
		\sum_{k = \floor*{\frac{m-1}{n_t}}n_t+1}^{m}
		\hspace*{-1.5ex}
		\left(\mathbf{L}_{\ell\left(\floor*{\frac{m-1}{n_t}} + 1\right)}\right)_
		{m - \floor*{\frac{m-1}{n_t}}n_t,k-\floor*{\frac{m-1}{n_t}}n_t}
		\hspace*{-1ex}
		x_{k}}^2
\end{multline*}
We then have (see \cite{My-Thesis} for details) 
\begin{align*}
&\hphantom{=} C(x_1x_2\cdots x_{LTn_t})\\
&= \sum_{\ell=1}^L
\lVert Y_\ell - \mathbf{L}_\ell X_\ell \rVert_\mathsf{F}^2 \\
&= \sum_{m = 1}^{LTn_t}
g_m
(x_{\floor*{\frac{m-1}{n_t}}n_t+1}
x_{\floor*{\frac{m-1}{n_t}}n_t+2}\cdots \,x_{m})\\
&=
\sum_{m = 1}^{LTn_t}
f_m
\mleft(x_{1}
x_{2}\cdots \,x_{m}\mright)
\end{align*}
where $$f_m
\mleft(x_{1}
x_{2}\cdots \,x_{m}\mright) = 
g_m
(x_{\floor*{\frac{m-1}{n_t}}n_t+1}
x_{\floor*{\frac{m-1}{n_t}}n_t+2}\cdots \,x_{m})$$
for $m = 1,2,\dots,LTn_t$. 
Thus, we have a causal cost function
as required and the vanilla stack
decoder is fully specified.
In fact, we have something 
better than what 
we needed. Each function in the additive
decomposition of the cost function
depends on at most $n_t$ previous terms
rather than at most $LTn_t$ previous terms.

Alternatively, we can interpret this as
the \textit{effective channel matrix} being block
diagonal with lower-triangular blocks. 
In particular, we can express
the cost function as 
\begin{equation}\label{Equivalent-MIMO-Problem}
	C(x_1x_2\cdots x_{LTn_t})
	=
	\lVert \mathbf y - \mathbf{H}\mathbf{x}\rVert_2^2
\end{equation}
where (recalling that $n_r = n_t$) $\mathbf{x},\mathbf{y}\in \mathbb{C}^{n_tLT}$ 
are given by
\begin{align*}
	\mathbf x &= \begin{bmatrix}
	x_1 & x_2 & \cdots & x_{n_tLT}
	\end{bmatrix}^\intercal \\ 
	\mathbf y &= \begin{bmatrix}
	y_1 & y_2 & \cdots & y_{n_tLT}
	\end{bmatrix}^\intercal 
\end{align*}
and $\mathbf{H}\in \mathbb{C}^{n_tLT\times n_tLT}$ is given
by
\begin{equation}\label{Effective-Channel-Proposed}
	\mathbf{H}
	= 
\mqty[\dmat{\mathbf{L}_1,\ddots,\mathbf{L}_1,\mathbf{L}_2,\ddots,\mathbf{L}_2,\ddots,\mathbf{L}_L}]
\end{equation}
where each $\mathbf{L}_\ell \in \mathbb{C}^{n_t\times n_t}$ for $\ell=1,2,\dots,L$ occurs $T$ times
on the diagonal of the matrix. On the other hand, the general case
of a causal cost function depending on all $LTn_t$ past terms corresponds
to the matrix $\mathbf{H}$ simply being lower-triangular.

\subsubsection*{Comparison to the Decoding of Linear Dispersion Codes}

As previously mentioned,
the decoding problem for a linear dispersion
code can be converted into a standard MIMO detection
problem for an $n_tLT\times 1$ vector over the
input constellation with an $n_tLT \times n_tLT$
channel (assuming $n_r = n_t$).
In particular, suppose that we have 
an $L$-block $n_t\times T$ 
linear dispersion code as defined
in \eqref{LD-Code-Definition}
with dispersion matrices
$A_1,A_2,\dots,A_{n_tLT}\in \mathbb{C}^{n_t\times LT}$. 
Denote
by $\mathbf{a}_i^{(j)}$ the $j$th column
of the $i$th dispersion matrix so that 
\begin{equation*}
	A_i 
	= 
	\begin{bmatrix}
	\mathbf{a}_i^{(1)} & 
	\mathbf{a}_i^{(2)} & \cdots
	&
	\mathbf{a}_i^{(LT)} 
	\end{bmatrix}
\end{equation*}
for $i = 1,2,\dots,n_tLT$.
One can then verify that the
cost function 
can be 
placed into the form of \eqref{Equivalent-MIMO-Problem}
by taking $\mathcal{S} = \mathcal{A}_\mathsf{in}^{n_tLT}$ 
(which is in bijection with the code) and taking 
$\mathbf{H}\in \mathbb{C}^{n_tLT\times n_tLT}$
to be 
\begin{equation}\label{Effective-Channel-LD-Pre}
\mathbf{H}
= \rho
\begin{bmatrix}
H_1 \mathbf{a}_1^{(1)} 
& H_1 \mathbf{a}_2^{(1)}
& \cdots
& H_1 \mathbf{a}_{n_tLT}^{(1)} \\
H_1 \mathbf{a}_1^{(2)} 
& H_1 \mathbf{a}_2^{(2)}
& \cdots
& H_1 \mathbf{a}_{n_tLT}^{(2)} \\
\vdots & \vdots & \ddots & \vdots \\
H_1 \mathbf{a}_1^{(T)} 
& H_1 \mathbf{a}_2^{(T)}
& \cdots
& H_1 \mathbf{a}_{n_tLT}^{(T)} \\

H_2 \mathbf{a}_1^{(T+1)} 
& H_2 \mathbf{a}_2^{(T+1)}
& \cdots
& H_2 \mathbf{a}_{n_tLT}^{(T+1)} \\
H_2 \mathbf{a}_1^{(T+2)} 
& H_2 \mathbf{a}_2^{(T+2)}
& \cdots
& H_2 \mathbf{a}_{n_tLT}^{(T+2)} \\
\vdots & \vdots & \ddots & \vdots \\

H_L \mathbf{a}_1^{(LT)} 
& H_L \mathbf{a}_2^{(LT)}
& \cdots
& H_L \mathbf{a}_{n_tLT}^{(LT)}
\end{bmatrix}\text{.}
\end{equation}
Observe now that
the problem of minimizing 
\eqref{Equivalent-MIMO-Problem}
is itself a special case of our decoding
problem 
so that the algorithm just developed is readily 
applicable.
In particular, we substitute $1$ for $L$
and $T$ and substitute $n_tLT$ for $n_t$.
We can then lower-triangularize the channel 
as before to get a causal cost
function or equivalently, a new
effective channel
\begin{equation}\label{Effective-Channel-LD}
\mathbf{H} = \mathbf{L}_1 
\end{equation}
where $\mathbf{L}_1 \in \mathbb{C}^{n_tLT \times n_tLT}$
is lower-triangular. Thus, one key difference
between the decoding problems for the proposed 
codes and linear dispersion codes is in the
effective channels \eqref{Effective-Channel-Proposed}
for the proposed codes 
and \eqref{Effective-Channel-LD} for linear dispersion codes. 
Apart
from the obvious ways in which 
the block diagonal structure of 
\eqref{Effective-Channel-Proposed}
is better such as reduced cost of QL decompositions, 
we will see that there are many more significant advantages.
We also emphasize that this structural benefit
still holds in the case of single-block coding (i.e., $L = 1$).

On the other hand, since $\mathcal{S} = \mathcal{A}_\mathsf{in}^{n_tLT}$
for a linear dispersion code, we will have a 
tree structure as in Fig.~\ref{LD-Tree}.
The major advantage in the decoding of linear dispersion codes
is that the decoding can be performed over the input
constellation which will generally
be smaller than the constellation
for an equal bpcu rate SRB code.
Observe that Figures \ref{My-Tree} and \ref{LD-Tree}
represent codes of equal size and hence
equal bpcu rate
but the latter has a smaller outdegree.

\begin{figure}[t]
	\centering
	\begin{minipage}{0.45\columnwidth}
		\centering
		\includegraphics{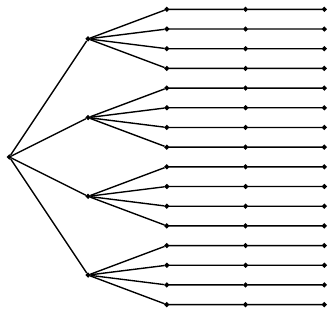}
		\caption{Structure of the code tree resulting from
			a systematic generator matrix}\label{My-Tree}
	\end{minipage}\hfill
	\begin{minipage}{0.45\columnwidth}
		\centering
		\includegraphics{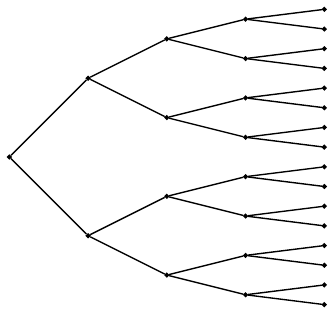}
		\caption{Structure of the code tree resulting from
			a linear dispersion code or uncoded
			signalling}\label{LD-Tree}
	\end{minipage}
\end{figure}

The key principle behind sequential
decoding is that of pruning
the tree of unpromising paths
via the priority measure
to avoid an exhaustive search.
This is easier to do in the 
case of a tree like that in 
Fig.~\ref{LD-Tree} rather than 
Fig.~\ref{My-Tree}.
In fact, for Fig.~\ref{My-Tree}, 
the length of the long legs
relative to the depth
up to which we have a full 
tree is proportional
to the diversity gain. In the
case of full diversity, the tree
constitutes only long legs stemming from the 
root node and no meaningful pruning is possible;
the vanilla stack decoder
becomes an exhaustive search.

The proposed
codes are thus inherently disadvantaged
from a tree structure perspective
and the closer they operate to full 
diversity, the closer the decoding
complexity becomes to exhaustive search.
Moreover, the tree structure is always worse
for any diversity gain greater than $1$ hence
for any nontrivial code.
With that said, we have seen that
they do not need to be close to full diversity
to perform well and we have the block diagonal
effective channel advantage which can 
compensate to some extent 
for the tree structure disadvantage.

In the coming subsections, we will proceed to provide
improvements upon the vanilla stack decoder. 
These improvements will typically 
involve some extra computations
or pre-computations that offer diminishing
returns as the SNR increases. 

\subsection{Future Costing}

We now consider an improvement
upon the vanilla stack 
decoder
via the A* algorithm \cite{A-star}. 
We adopt
some of the terminology associated with A*.
The idea is to add to the cost
associated with a prefix a lower 
bound on the cost of any possible 
continuation of that prefix---the future cost.
The hope is that this will cause more paths to
be pruned early on in the decoding procedure. 

Consider 
$n$ \textit{heuristic functions}
$h_i\colon \mathcal{A}^i \longrightarrow \mathbb{R}_{\geq 0}$
for $i = 1,2,\dots,n$.
Define the \textit{modified extended causal cost function}
$C'\colon \mathcal{S}^* \longrightarrow \mathbb{R}_{\geq 0}$ 
by 
\begin{equation*}
C'(a_1a_2\cdots a_i)
= C(a_1a_2\cdots a_i)
+ h_i(a_1a_2\cdots a_i)
\end{equation*}
for all  $a_1a_2\cdots a_i\in\mathcal{S}^*$  
(with $a_1,a_2,\dots,a_i\in\mathcal{A}$)
for $i = 1,2,\dots,n$.
The heuristic functions 
are said to be \textit{admissible} if
\begin{equation}\label{Admissible-Heuristic}
h_i(a_1\cdots a_i)
\leq \hspace*{-1.75ex}
\min_{
	\substack{b_1\cdots b_n \in \mathcal{S}\\
		b_1\cdots b_i = a_1\cdots a_i}} \hspace*{-1.75ex}
f_{i+1}(b_1\cdots b_{i+1}) + \cdots + f_{n}(b_1\cdots b_{n})
\end{equation}
for all  $a_1a_2\cdots a_i\in\mathcal{S}^*$  
(with $a_1,a_2,\dots,a_i\in\mathcal{A}$)
for $i = 1,2,\dots,n-1$
and 
\begin{equation*}
h_n(a_1\cdots a_n) = 0
\end{equation*}
for all $a_1\cdots a_n \in\mathcal{S}$.

We claim that if the cost function in
the vanilla stack decoder is replaced 
with a modified extended causal cost function
with admissible heuristics, the algorithm
still terminates and yields the minimum cost codeword.
For a proof of this, again, the reader may simply
recognize this as an instance of A* \cite{A-star}
or refer to \cite{My-Thesis}.

\subsubsection*{Application to the Decoding of the Proposed Codes}

The utility of a heuristic relies critically on the complexity
of its evaluation. Solving the minimization
in \eqref{Admissible-Heuristic} for every 
$a_1 a_2 \cdots a_i \in \mathcal{S}^*$ for $i = 1,2,\dots,n$
is just as hard as the decoding problem that we are trying
to solve to start with. We instead hope to find
lower bounds on that minimization that are sufficiently
easy to compute. 
The particular topic of obtaining
such bounds in the context of sphere decoding is
studied extensively in \cite{SDP-Inspired-Lower-Bounds}.
The same topic is studied in \cite{Eigenbound-Equivalent}
in the context of MIMO detection with an A* algorithm 
and the authors arrive at a bound equivalent to one appearing
in \cite{SDP-Inspired-Lower-Bounds}. The techniques
from these works are more or less 
readily applicable to our decoder.
In particular,
one might need to relax the constellation constraints
to $\Lambda = \mathbb{G}$ or $\Lambda = \mathbb{E}$
in which case the appropriate
quantization functions can be applied to get
constant-time solutions to $1$-dimensional detection problems.
This has the effect of slightly loosening the resultant bounds.

However, there is one way in which our
problem differs from those considered
in \cite{SDP-Inspired-Lower-Bounds,Eigenbound-Equivalent}
that we can exploit. This is in that we have
a block diagonal lower-triangular effective channel matrix
\eqref{Effective-Channel-Proposed}
rather than a dense lower-triangular effective
channel matrix \eqref{Effective-Channel-LD}.
This gives the proposed
codes a decoding advantage
in the particular area of A* 
heuristics. We interject to elaborate.

\subsubsection*{Comparison to the Decoding of Linear Dispersion Codes}

Earlier, we showed that our cost function
additively decomposes into functions which depend on 
at most the past $n_t$ terms
$$f_m
\mleft(x_{1}
x_{2}\cdots \,x_{m}\mright) = 
g_m
(x_{\floor*{\frac{m-1}{n_t}}n_t+1}
x_{\floor*{\frac{m-1}{n_t}}n_t+2}\cdots \,x_{m})$$
for $m = 1,2,\dots,LTn_t$.

On the other hand, in the decoding of linear
dispersion codes, it is neither feasible nor intended that
the codeword be considered directly in the decoding process
due to the enormous constellation size. Rather, they are 
intended to be decoded over their \textit{input} constellation. 
For this to be done, 
the channel matrices must be combined with the generator
matrices leading to an $n_tLT\times n_tLT$ effective 
channel \eqref{Effective-Channel-LD-Pre} acting
on a single $n_tLT \times 1$ vector with symbols from the input
constellation. After lower-triangularizing the channel, 
the cost function 
thus ends up taking the most general 
form \eqref{Causal-Cost-Function}
with dependence
on the past $n_tLT$ symbols. 

The implication this has for heuristics is that they need
to be constantly re-computed on the fly as nodes are explored 
since the minimization depends on the previous symbols. 
This strongly restricts the complexity allowed
for these heuristics. Naive choices lead to complexity
equivalent to that of simply visiting the nodes that we 
are trying to prune.
In contrast, if the cost function breaks up
as it does for the proposed codes, 
we get heuristics
that can be pre-computed and are identical for any
prefix. For example, we can solve the column-by-column
MIMO detection problem or find a lower bound on its cost
and this only has to be done once. The fact that it only
has to be done once means that heuristics that would
otherwise be useless, like solving a lower-dimensional
version of the same decoding problem under consideration,
are now potentially useful. 

\subsubsection*{Application to the Decoding of the Proposed Codes Continued}

We require that our heuristic satisfies
\begin{equation}\label{Heuristic-Bound}
h_i(x_1\cdots x_i)
\leq 
\min_{
	\substack{z_1\cdots z_n \in \mathcal{S}\\
		z_1\cdots z_i = x_1\cdots x_i}} 
\sum_{j = i+1}^{LTn_t} f_{j}(z_1\cdots z_{j})
\end{equation}
for $i = 1,2,\dots,LTn_t - 1$
and $h_{LTn_t}(x_1\cdots x_{LTn_t}) = 0$.
As previously mentioned, such heuristics
can be found in \cite{SDP-Inspired-Lower-Bounds,Eigenbound-Equivalent}.
However, we will shift our interest
to heuristics that do not
depend on $x_1\cdots x_i$ in which
case we can use more complex heuristics (in 
the sense of computational complexity), that
are conceptually trivial. In particular,
simply solving smaller versions of the
decoding problem under consideration.

Firstly, we can remove the codeword
constraint for the purposes of the lower-bounding
\eqref{Heuristic-Bound}. That is, 
\begin{align*}
&\hphantom{\geq}\min_{
	\substack{z_1\cdots z_n \in \mathcal{S}\\
		z_1\cdots z_i = x_1\cdots x_i}} 
\sum_{j = i+1}^{LTn_t} f_{j}(z_1\cdots z_{j})\\
&\geq 
\min_{
	\substack{z_1\cdots z_n \in \mathcal{A}^{LTn_t}\\
		z_1\cdots z_i = x_1\cdots x_i}} 
\sum_{j = i+1}^{LTn_t} f_{j}(z_1\cdots z_{j})\text{.}
\end{align*}
We further would like that our heuristic does
not depend on $x_1\cdots x_i$. We can 
remove the appropriate number of terms
from the beginning of the sum:
\begin{align}
\sum_{j = i+1}^{LTn_t} f_{j}(z_1\cdots z_{j})
&\geq 
\sum_{j = \floor*{\frac{i-1}{n_t}}n_t+n_t+1}^{LTn_t}
f_{j}(z_1\cdots z_{j})\nonumber\\
&= 
\sum_{j = \floor*{\frac{i-1}{n_t}}n_t+n_t+1}^{LTn_t}
g_{j}(z_{\floor*{\frac{j-1}{n_t}}n_t+1}\cdots\, z_{j})\label{Future-Cost}
\end{align}
where we have $\floor*{\frac{i-1}{n_t}}n_t+n_t+1 \geq i + 1$
and $\eqref{Future-Cost}$ depends on 
$z_{\floor*{\frac{i-1}{n_t}}n_t+n_t+1}\cdots z_{LTn_t}$.

There are three points that must be noted now.
The first is that we can break up the sum 
\eqref{Future-Cost} in any way we like and 
lower-bound its minimization by the term-by-term
minimization of the broken up sum. The second 
is that we can lower-bound any terms
trivially by $0$. This allows us to obtain
lower-bounds which consist of MIMO detection problems
of any size, i.e, with $1 \times 1$ to $n_t \times n_t$
channels. The third is that these do
not need to be computed for every $i = 1,2,\dots, LTn_t-1$
(this a separate matter from independence of
$x_1\cdots x_i$).
The same heuristic will be shared by ranges of $i$ and
the heuristics for larger $i$ just involve fewer terms
of \eqref{Future-Cost} so are obtained by subtracting 
components of the heuristic for smaller $i$.

For example, for $i = 1$, we
can break down \eqref{Future-Cost}
into $LT - 1$ terms
corresponding to
the last $LT-1$ columns
of the codeword $X$:
\begin{align} 
&\sum_{j = n_t + 1}^{2n_t} 
g_{j}(z_{\floor*{\frac{j-1}{n_t}}n_t+1}\cdots\, z_{j}) 
+ \hspace*{-1.2ex}
\sum_{j = 2n_t + 1}^{3n_t} 
g_{j}(z_{\floor*{\frac{j-1}{n_t}}n_t+1}\cdots\, z_{j}) \nonumber\\*
& + 
 \cdots + 
\sum_{j = (LT-1)n_t + 1}^{LTn_t} 
g_{j}(z_{\floor*{\frac{j-1}{n_t}}n_t+1}\cdots\, z_{j})
\label{VBT-Analogue}\text{.}
\end{align}
Minimizing each of these constitutes a standard MIMO detection
problem with an $n_t\times n_t$ channel. These
can be solved by another instance of the decoder
we are describing---perhaps a more vanilla one
to give us $LT-1$ minima. 
The heuristic for $i = 1,2,\dots,n_t$ would
be the sum of these $LT - 1$ minima. 
For $i = n_t+1,\dots,2n_t$, it would be 
the sum of the last $LT-2$ of these minima, and
so on.
This is in direct analogy to obtaining
the cost of symbol-by-symbol hard decisions
in the VBT metric decoder described
in \cite{VBT-Metric}. In fact, with $n_t = 1$
it would be precisely that but with a fading 
channel.

As previously
mentioned, we need not solve an $n_t\times n_t$ problem,
we can solve a $1 \times 1$ problem 
which would correspond
to replacing \eqref{VBT-Analogue} by the
first term in each sum and costs almost nothing. 
The SNR and the codebook
size will determine whether the computational
effort put into computing the heuristic
is made up for by a reduction of the complexity
of the main tree search.
For example, in obtaining the $n_t = 4,\,L = 1$
simulation result in the previous section,
it was found that solving the $n_t\times n_t$
problem was worthwhile and significantly reduced
the overall complexity. 

Finally,
note that we can mix and match
heuristics such as those described here
and those intended for dependence on $x_1\cdots x_i$.
Generally, if our codeword $X$ is short and wide, 
we can get sufficiently tight heuristics which do not
depend on the previous symbols, but if the codeword
is tall and narrow, it might only be worthwhile
to consider heuristics which depend on $x_1\cdots x_i$
and simply apply the approaches in
\cite{SDP-Inspired-Lower-Bounds,Eigenbound-Equivalent}.

Moreover, completely separately from this point,
we can
use the computationally less
expensive heuristics in
\cite{SDP-Inspired-Lower-Bounds,Eigenbound-Equivalent}
to lower-bound the terms of \eqref{VBT-Analogue}
rather than outright solve the minimizations.
For example, we can consider the bound in \cite{SDP-Inspired-Lower-Bounds}
termed the \textit{eigenbound} which is also independently
arrived at in \cite{Eigenbound-Equivalent}. Take some
$M \in \mathbb{C}^{n_t\times n_t}$ 
and $\mathbf{u},\mathbf{v} \in \mathbb{C}^{n_t}$
with $M$ invertible
and let $\lambda_\mathsf{min}$ be the smallest
eigenvalue of the positive definite matrix
$M^\dagger M$.
Then, one can easily show that \cite{SDP-Inspired-Lower-Bounds}
\begin{equation*}
\Vert M\mathbf{u} - \mathbf{v}\rVert_2^2
\geq 
\lambda_\mathsf{min}
\Vert \mathbf{u} - M^{-1}\mathbf{v}\rVert_2^2\text{.}
\end{equation*}
We can then perform the minimization over $\mathbf{u}\in \mathcal{A}^{n_t}$
component-wise. We can simplify this further by relaxing the
minimization to over $\mathbf{u}\in \Lambda^{n_t}$
with $\Lambda = \mathbb{G}$ or $\mathbb{E}$ and use the 
appropriate quantization function.

\subsection{Spherical Bounding}

We consider using a spherical bound stack decoder
first proposed in \cite{Spherical-Bound-Stack-Decoder}.
This decoder is essentially a best-first 
variation on the depth-first 
sphere decoder \cite{Sphere-Decoder,Sphere-Decoder-2}
which is a commonly used decoder for linear
dispersion codes.
The spherical bounding 
idea is to pick a threshold $\mathcal{T}$
and restrict the search to the $s\in\mathcal{S}$
satisfying $C(s) \leq \mathcal{T}$. If there 
are no codewords satisfying $C(s) \leq \mathcal{T}$,
increase $\mathcal{T}$ and start over. Clearly, 
if we find the minimum cost codeword among
the codewords whose cost is less than or equal
to $\mathcal{T}$, then it must be the
minimum cost codeword among all codewords
so the resulting decoder is still ML.

The modification to the basic algorithm is 
simply as follows: Rather than pushing $pt$
onto the stack for every $t \in E(p)$,
we push $pb$ onto the stack for every
\begin{equation}\label{SB-Set}
b \in 
\{t\in E(p) \mid f_{\abs{p}+1}(pt) \leq \mathcal{T}-C(p)\}\text{.}
\end{equation}
There is now a possibility that the priority queue size 
will 
decrease because something popped might not be replaced.
Nonetheless, one can simply increase $\mathcal{T}$ and start
over should the priority queue become empty. It
can again be verified that the resulting algorithm
works. Moreover, combining
this with future costing is also straightforward.

The challenge now 
is in enumerating the elements 
of the set
in \eqref{SB-Set}.
In the case of $\abs{p} < kn_t$, we have $E(p) = \mathcal{A}$
so we must find the constellation points $t \in \mathcal{A}$
for which $f_{\abs{p}+1}(pt) \leq \mathcal{T}-C(p)$.
The naive way would be to go through all of 
the constellation
points and compare to the threshold. However, this
would give us only a space complexity reduction. 
In order to obtain a time complexity reduction,
this must be done
without going through all of the constellation points every time.
We will
start by showing that this problem amounts to finding
the constellation points that lie in a circle
in the complex plane. 

Let $p = x_1\dots x_\abs{p}$,
define $h$ by
\begin{equation}\label{Channel-Coefficient}
h = 
\left(\mathbf{L}_{\ell\left(\floor*{\frac{\abs{p}}{n_t}} + 1\right)}\right)_
{\abs{p} + 1
	- \floor*{\frac{\abs{p}}{n_t}}n_t,
	\abs{p}+1-\floor*{\frac{\abs{p}}{n_t}}n_t}\text{,}
\end{equation} 
and define $u$ by
\begin{multline*}
	u =\\ y_{\abs{p}+1} - \hspace*{-1ex}
	\sum_{k = \floor*{\frac{\abs{p}}{n_t}}n_t+1}^{\abs{p}}\hspace*{-1ex}
	\left(\mathbf{L}_{\ell\left(\floor*{\frac{\abs{p}}{n_t}} + 1\right)}\right)_
	{\abs{p} + 1 - \floor*{\frac{\abs{p}}{n_t}}n_t,k-\floor*{\frac{\abs{p}}{n_t}}n_t}
	x_{k}\text{.}
\end{multline*}
One can then verify that 
\begin{equation*}
f_{\abs{p}+1}(pt)=
\abs{u-ht}^2
\end{equation*}
which leads to 
\begin{equation*}
\abs{u-ht}^2 \leq \mathcal{T} - C(p)\text{.}
\end{equation*}
This is equivalent to 
\begin{equation}\label{Circle-Constraint}
\abs{t - c} \leq r
\end{equation}
where
\begin{equation}\label{Radius}
r =  \frac{\sqrt{\mathcal{T}-C(p)}}{\abs{h}}
\end{equation}
and 
\begin{equation*}
c = \frac{u}{h}\text{.}
\end{equation*}
We seek to enumerate
the constellation points $t\in \mathcal{A}$
satisfying \eqref{Circle-Constraint}. At
this point, we interject to comment on what
happens in the case of linear dispersion codes.

\subsubsection*{Comparison to the Decoding of Linear Dispersion Codes}

When a linear dispersion code is used with 
a QAM input constellation which is a Cartesian product
of PAM constellations, 
the channel model is usually
converted into an equivalent real-valued model. 
Finding the constellation points
which satisfy \eqref{Circle-Constraint} reduces
to the trivial problem of enumerating the integers 
on an interval. For this reason, most of the literature
concerned with the decoding of linear dispersion codes
works with real-valued models. In contrast, we do 
not have that option and cannot readily apply 
the existing spherical bounding procedures.

\subsubsection*{Spherical Bounding Continued}

We will now provide procedures
for finding Gaussian and Eisenstein 
integer constellation points in a circle
without necessarily exhaustively
going through all $\abs{\mathcal{A}}$ points.
The procedures will be simple 
and applicable to arbitrary Gaussian
or Eisenstein integer constellations. 
They 
might not necessarily
be the most efficient possible procedures, but 
they will be efficient enough to realize
the task at hand which is obtaining the constellation points
satisfying \eqref{Circle-Constraint} in a small fraction
of $\abs{\mathcal{A}}$ steps (averaging over high SNRs).
Figures \ref{Gaussian-Rejection} and \ref{Eisenstein-Rejection}
provide visualizations of these procedures. 

We start by considering the problem
of finding the points in $\Lambda$ which are
inside the circle
where $\Lambda = \mathbb{G}$ or $\Lambda = \mathbb{E}$.
Denote by $R^\mathsf{circle}_\Lambda(c,r)$ the set of 
lattice points inside the circle defined by \eqref{Circle-Constraint}, i.e.,
\begin{equation*}
R^\mathsf{circle}_\Lambda(c,r)
= 
\{z \in \Lambda \mid \abs{z-c}\leq r\}\text{.}
\end{equation*}

\begin{figure}[t]
	\centering
	\includegraphics[width=0.6\columnwidth]{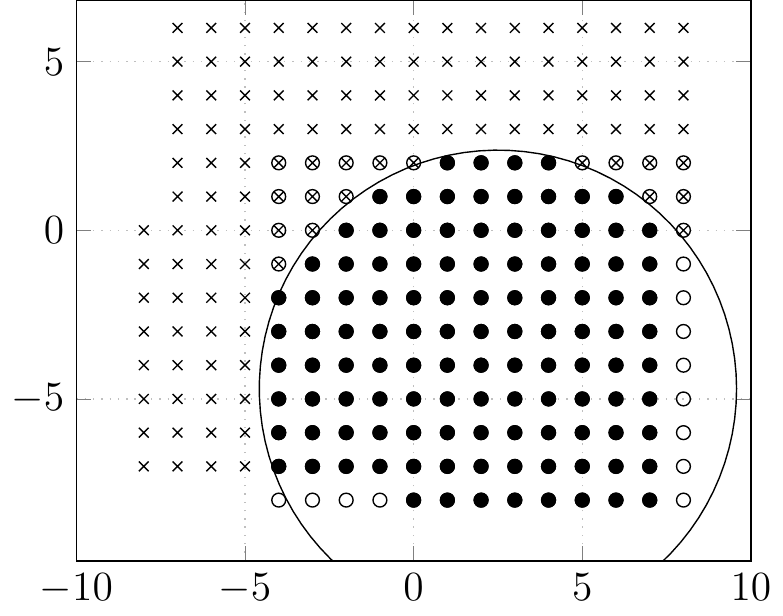}
	\caption{Finding Gaussian integer constellation points in a circle; the hollow
		circles are the Gaussian integers that are visited}\label{Gaussian-Rejection}
	\vspace{1cm}
	\includegraphics[width=0.6\columnwidth]{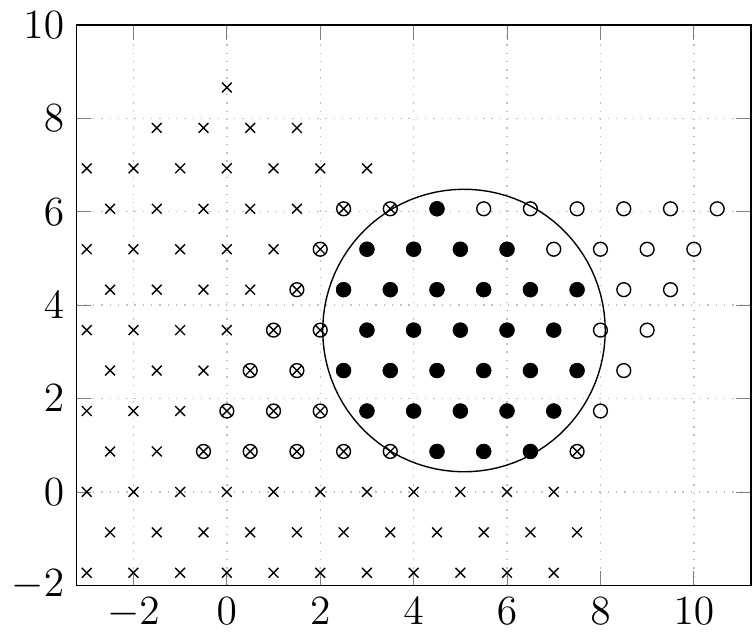}
	\caption{Finding Eisenstein integer constellation points in a circle;
		the hollow circles are the Eisenstein integers that are visited}\label{Eisenstein-Rejection}
\end{figure}

We start with the case of $\Lambda = \mathbb{G}$.
Rectangular subsets of $\mathbb{G}$ can be naturally enumerated.
We can start by finding the smallest square which contains the
circle and then \textit{shrinking} it so that the edges
are aligned with $\mathbb{G}$. Let 
\begin{align*}
i_\mathsf{min} &= \ceil*{\Re{c} - r}\\
i_\mathsf{max} &= \floor*{\Re{c} + r}\\
j_\mathsf{min} &= \ceil*{\Im{c} - r}\\
j_\mathsf{max} &= \floor*{\Im{c} + r}
\end{align*}
and denote by $R_\mathbb{G}^\mathsf{rectangle}(c,r)$ the set
\begin{multline*}
R_\mathbb{G}^\mathsf{rectangle}(c,r) = \\
\{i + \imath j \mid i,j \in \mathbb{Z},\, 
i_\mathsf{min} \leq i \leq i_\mathsf{max},\,
j_\mathsf{min} \leq j \leq j_\mathsf{max}\}\text{.}
\end{multline*}
One can verify that $R_\mathbb{G}^\mathsf{circle}(c,r)
\subseteq R_\mathbb{G}^\mathsf{rectangle}(c,r)$. 
We can then enumerate the elements of 
$R_\mathbb{G}^\mathsf{circle}(c,r)$ by enumerating
the elements of $R_\mathbb{G}^\mathsf{rectangle}(c,r)$
and rejecting the $z \in \mathbb{G}$ for which
$\abs{z-c} > r$.

In the case of $\Lambda = \mathbb{E}$, parallelogram regions
are naturally enumerated. 
We can similarly find a parallelogram which contains
the circle and align it with $\mathbb{E}$ in a
manner that guarantees that no points inside the circle are missed. 
Denote by $R_\mathbb{\mathbb{E}}^\mathsf{parallelogram}(c,r)$ the set
\begin{multline*}
R_\mathbb{E}^\mathsf{parallelogram}(c,r) = \\
\left\{i_0 + j_0 \frac{\imath\sqrt{3}}{2} + i + j(\omega+1) \mmiddle{|}
\begin{array}{c}
i,j \in \mathbb{Z} \\ 
0 \leq i \leq i_\mathsf{max}\\
0 \leq j \leq j_\mathsf{max}
\end{array}
\right\}
\end{multline*}
where $i_0$, $j_0$, $i_\mathsf{max}$, and $j_\mathsf{max}$ are
as follows:
\begin{align*}
j_0 &= \ceil*{\frac{\Im{c} - r}{\frac{\sqrt{3}}{2}}}\\
i_0 &= 
\begin{cases}
\floor*{\Re{c} - r\sqrt{3}} & \text{if $j_0$ is even}\\
\floor*{\Re{c} - r\sqrt{3} - 0.5}+0.5 & \text{if $j_0$ is odd}
\end{cases}\\
j_\mathsf{max} &= 
\floor*{\frac{\Im{c}+r}{\frac{\sqrt{3}}{2}}} - j_0\\
i_\mathsf{max}
&= \\
&
\hspace*{-3ex}
\begin{cases}
\ceil*{\Re{c} + r\sqrt{3}} -
\Re{\xi}& \text{if $\frac{\Im{\xi}}{\frac{\sqrt{3}}{2}}$ is even}\\
\ceil*{\Re{c} + r\sqrt{3} - 0.5} + 0.5 -
\Re{\xi}& \text{if $\frac{\Im{\xi}}{\frac{\sqrt{3}}{2}}$ is odd}
\end{cases}\\
&\quad\text{where $\xi = i_0 + j_0 \frac{\imath\sqrt{3}}{2} + j_\mathsf{max}(\omega+1)$.}
\end{align*}
It can be verified
that $R_\mathbb{\mathbb{E}}^\mathsf{circle}(c,r)
\subseteq R_\mathbb{\mathbb{E}}^\mathsf{parallelogram}(c,r)$. 
As before, we can enumerate the elements of 
$R_\mathbb{G}^\mathsf{circle}(c,r)$ by enumerating
the elements of $R_\mathbb{E}^\mathsf{parallelogram}(c,r)$
and rejecting the $z \in \mathbb{E}$ for which
$\abs{z-c} > r$.

The next task is to restrict the
points in $R^\mathsf{circle}_\Lambda(c,r)$
to the 
points that belong to our constellation $\mathcal{A} = \mathcal{A}_{\Pi\Lambda}$.
Obviously, it would defeat the purpose of the algorithm if we 
were to have to
compare them to each point in $\mathcal{A}$. Fortunately, 
Proposition \ref{Set-Membership-Remark} gives us a constant-time deterministic
set membership test for checking whether some $z \in \Lambda$ belongs
to  $\mathcal{A}_{\Pi\Lambda}$. 
We can find $\mathcal{A}_{\Pi\Lambda} \cap R^\mathsf{circle}_\Lambda(c,r)$
by simply rejecting the $z \in \Lambda$ for which $Q_\Lambda(z/\Pi) \neq 0$
during the enumeration.

One issue remains. When the
SNR is low, the radius $r$
is frequently large enough
for the circle to cover 
an area larger than the
area covered by the constellation
aside from the possibility of the circle 
being entirely outside of the
constellation.
Consequently, at low SNRs, 
this algorithm would lead to 
a higher average complexity
than that of exhaustively going through all of 
the points in $\mathcal{A}$.
Therefore, we must add a further step which is the 
enforcement of the boundaries of the constellation
so that the worst-case complexity is comparable to
enumeration by exhaustive search.
We propose that this be done in simplest possible way which is by
enforcing a rectangular boundary around the constellation.
This is achieved by
modifying the parameters of the rectangle and parallelogram 
regions to
prevent them from extending further outside 
this rectangular
boundary enclosing the constellation than is
necessary. 
In the case of 
rectangle regions, this is trivial. 
In the case of parallelogram regions,
it is a matter of elementary
geometry which we will omit to detail.
Further refinement of the procedure offers diminishing
returns
because if the SNR is low enough
for the circles to consistently stretch
far beyond the constellation, there would 
not be any benefit to spherical bounding to
begin with.

For completeness, we specify an
initial choice of threshold. We start with
\begin{equation*}
\mathcal{T}  = \alpha \cdot \sum_{\ell=1}^L \Expect\mleft[\lVert W_\ell \rVert_\mathsf{F}^2\mright]
\end{equation*}
where the initial value of $\alpha$ is a manually tuned parameter.
Whenever the search fails, we add $\delta$ to $\alpha$
where $\delta$ is also a manually tuned parameter. 
Justification for such a form of threshold can be found in \cite{STWS-Ch15}.

\subsection{Permutations}

\subsubsection*{Spatial Permutations}

We have some freedom in permuting the rows of
the sub-codewords of $X$ and hence
the order in which the symbols within
a column are detected. 
In particular, 
let $P_1,P_2,\dots,P_L$ be some $n_t\times n_t$ permutation
matrices. We then have
$\rho H_\ell X_\ell = (\rho H_\ell P_\ell) (P_\ell^{-1} X_\ell)$.
We can take $\rho H_\ell P_\ell$ as our channel matrix
and perform the QL decomposition $\rho H_\ell P_\ell = Q_\ell \mathbf{L}_\ell$
for $\ell=1,2,\dots,L$.  We refer to such permutations 
as \textit{spatial permutations.} 
For the purposes of decoding our codes, 
the modification to the algorithm
is straightforward: We can perform the decoding
as usual but must undo the permutation when 
computing the parity symbols as well as at the end of the decoding procedure.  

The topic of spatial permutations is studied extensively 
in \cite{ML-CLPS,Tree-Search-Decoding}. The goal is usually
to choose a permutation which leads to the matrices 
$\mathbf{L}_1,\mathbf{L}_2,\dots,\mathbf{L}_L$
having properties which result in a more efficient tree search. For example,
from \eqref{Channel-Coefficient} and \eqref{Radius}, we see that having
large magnitude coefficients on the diagonals of these matrices
leads to smaller radii in the spherical bounding procedure. 
We will consider the use of a 
simple heuristic proposed in \cite{ML-CLPS} which
is to sort the columns of the channel matrices in
descending order
of $2$-norm.

\subsubsection*{Temporal Permutations}

We have the freedom to permute the columns of $X$ arbitrarily
and hence detect the columns of the codeword in any order.
This is enabled by the fact that the underlying code is MDS. 
Once a permutation has been chosen, the modification to the algorithm
is straightforward: 
\begin{itemize}
	\item Permute the columns of $Y$ accordingly as well as their
	associations with the different channel matrices (via $\ell(j)$). 
	\item Permute 
	the columns of the generator matrix used for generating
	the code tree accordingly and systematize it. 
	\item Decode as usual and apply the inverse permutation
	to the columns of the resulting codeword.
\end{itemize}

It remains to determine how to choose a temporal permutation.
We propose that we use the permutation which puts
the columns of $Y$ into descending 
order of $2$-norm as is consistent with the principle
of detecting the most reliable parts of the codeword first.

\subsubsection*{Comparison to the Decoding of Linear Dispersion Codes}

In the case of linear dispersion codes, 
the effective codeword
is an $n_tLT\times 1$ vector and there is no option for temporal
permutations, we can only do spatial permutations. 
Moreover, since we
have one large $n_t LT \times n_t LT$ effective channel matrix
rather than several $n_t \times n_t$ matrices, the number
of possible spatial permutations is larger and the choice 
among them is important as is illustrated in 
\cite{ML-CLPS,Tree-Search-Decoding,Advanced-Spatial}.

It is worth noting that 
\cite{Advanced-Spatial} argues
that \textit{spatial} permutations
should be chosen as a function of both the
channel matrix \textit{and} the received signal vector
and provides a relatively complex method for doing so.
On the other hand, for the proposed codes, the block diagonal
structure allows us to use the received signal
to influence the detection order trivially via
the proposed temporal permutation.

\subsection*{The General Case}

We have developed an algorithm and a variety
of improvements that are readily applicable to the decoding
of SRB codes in the case of $n_r = n_t$. 
The case of $n_r > n_t$ is handled in the same way 
with the QL decomposition for rectangular matrices applied.
This results in an immaterially different cost function.
The reader is referred to \cite{ML-CLPS,Tree-Search-Decoding}
for details. 

On the other hand, the case of $n_r < n_t$ 
involves some challenges. 
A sphere decoder
for the case of $n_r < n_t$ is proposed in \cite{Asymmetric-Sphere-Decoder}
and the reader is also referred to \cite{ML-CLPS} for some comments on this
matter.
In this case,
the first term in the causal cost function
depends on $n_r-n_t$ symbols rather than one. 
To use a stack decoder, we must begin by pushing all prefixes
of length $n_r-n_t$ onto the priority queue incurring exponential
complexity in $n_r - n_t$. Depth-first decoders
might be more natural in this case.
Nonetheless, the challenge of $n_r > n_t$
is not unique to the proposed codes 
and the techniques 
we provided for spherical bounding can be used
to develop depth-first strategies.

Finally, we have the topic of SRA codes. In the 
case where there is no sub-codeword which contains 
both parity and information rows, all of the same decoding
techniques apply readily. When this is not the case,
it is not immediately apparent how 
to appropriately enumerate the codewords,
(i.e., compute the $E$ function \eqref{E-Function}).
 A standard but not necessarily very 
efficient way of dealing with this
is by removing the codeword constraints and checking
at the end of the decoding procedure if the answer
is a valid codeword. If the answer is not a valid
codeword, the search is continued.
Checking whether the codeword is valid can be done 
efficiently
by representing the code as the null space of a 
\textit{parity-check matrix} rather than the 
row space of a generator matrix.

\subsection{Decoding Complexity Simulations}\label{Decoding-Complexity-Simulations}

In this subsection, we examine the relative complexities of the proposed
decoding strategies as a function of SNR in simulation.  We consider the
case of $n_r = n_t = T = L = 2$ and a $d = 3$ SRB code with a $271$-Eis.\
constellation achieving a bpcu rate of $8.08$.  The CER versus SNR curve
for this code was provided in Section \ref{2-Block-Sim-Sec} in
Fig.~\ref{2-Block-CERs-Big-Constellations}.

The spherical bounding parameters are $\alpha = 1.75$ and $\delta = 0.25$.
Moreover, the future costing is done by using the eigenbound to lower-bound
the cost of the column-by-column MIMO detection problem where the
component-wise minimization is performed with constellation constraints.

Fig.~\ref{Time} plots the average number of code tree nodes visited as a
function of SNR which is a proxy for the time complexity. In counting the
number of nodes visited, we count all of the Eisenstein integers in the
parallelogram occurring in the spherical bounding procedure as visited
nodes even though some of them are not constellation points and are,
strictly speaking, not part of the code tree. Nonetheless, they effectively
act like virtual code tree nodes that are visited and rejected.  Thus, the
inefficiency of relaxing the circle constraint to a parallelogram
constraint is accounted for. Moreover, the visit count is not reset if the
tree search starts over due to the bounding threshold being too small.

\begin{figure}[t]
	\centering
	\includegraphics[width=\columnwidth]{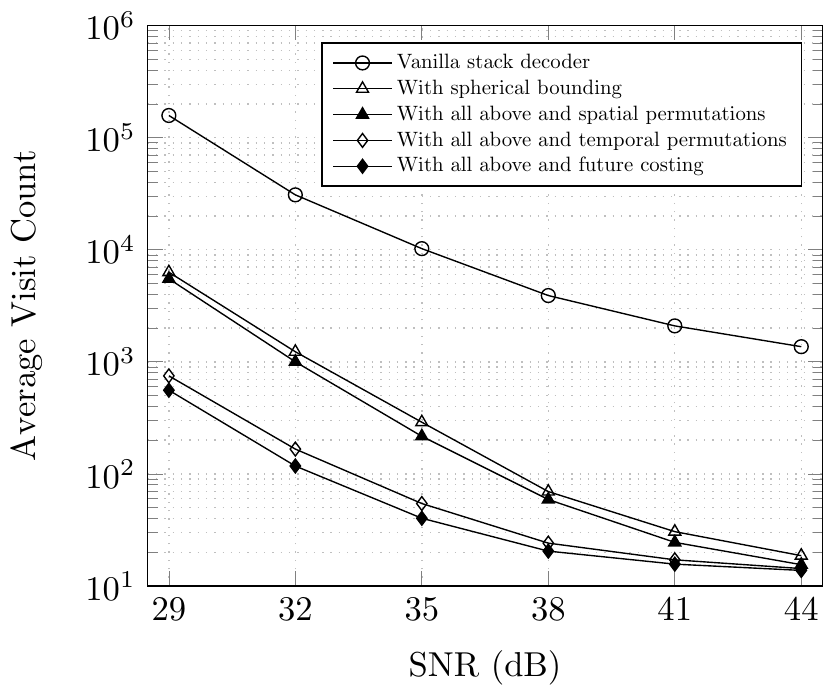}
	\caption{Average number of code tree nodes visited versus SNR}\label{Time}
\end{figure}
\begin{figure}[t]
	\centering
	\includegraphics[width=\columnwidth]{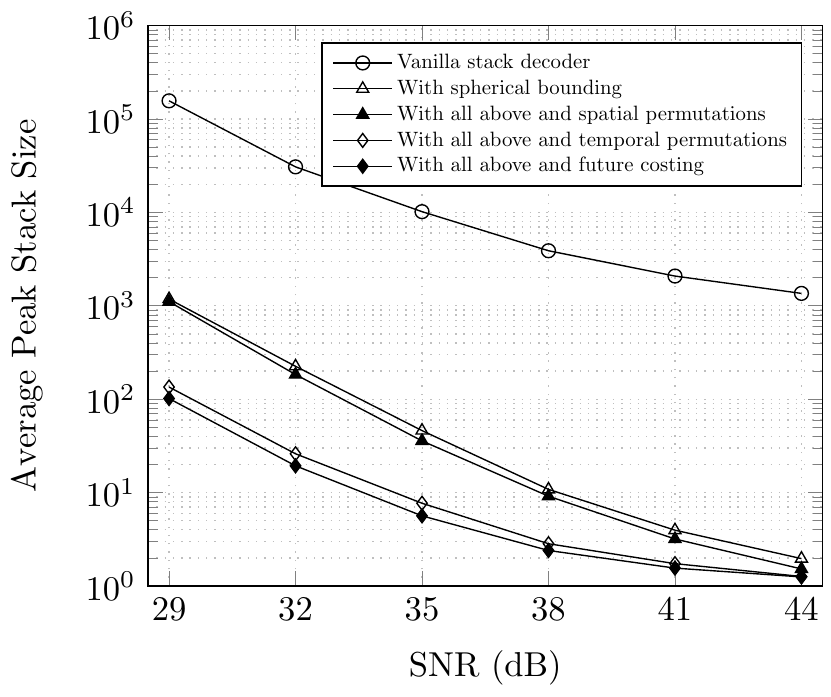}
	\caption{Average peak stack size versus SNR}\label{Space}
\end{figure}

Fig.~\ref{Space} plots the average peak stack (or priority queue) size as a
function of SNR which is a proxy for the space complexity. The
\textit{peak} is referring to the fact that the tree search might be
repeated several times due to too small a bounding threshold.

The results are self-explanatory so our comments will be brief.  Firstly,
note that the simulations start at an SNR which is already quite
high---corresponding to a CER of less than $10^{-4}$. The reason is that
the vanilla stack decoder is infeasible at lower SNRs.  One can expect
larger gaps between the successive improvements at lower SNRs. Moreover,
this is only one example and the relative significance of the various
decoder improvements depends on the channel, code, and decoder parameters.
For example, complex future costing heuristics were crucial to the decoding
of the $n_t = 4,L = 1$ code simulated in Section \ref{1-Block-Sim-Sec}
whereas here, they appear to be insignificant because a simple loose bound
is used.

\section{Concluding Remarks}\label{Conclusion-Sec}

In this paper, we studied space--time codes based on rank and sum-rank
metric codes from both the perspectives of of coding-theoretic optimality
properties and empirical error performance.  We demonstrated that such
codes can have competitive performance relative to codes designed for other
criteria aside from utilizing significantly smaller constellations.
Moreover, we demonstrated that such codes can be feasibly decoded with new
challenges and opportunities arising from the decoding problem.  Apart from
the obvious ways in which these investigations can be extended, we suggest
the following broad questions for future work:

1) Suboptimal decoding was not considered in this paper and the decoding
complexity remains relatively high.  It could be interesting to consider
whether the lattice reduction approach of Puchinger et al. in \cite{Sven}
could be usefully combined with the proposed sequential decoding
strategies.

2) We have demonstrated that the proposed codes can sometimes outperform
full diversity codes even at higher SNRs. Yet, the rate--diversity
optimality criterion for which the codes were constructed is not
necessarily of any relevance to this error performance.  In light of this,
it could be interesting to consider other ways of designing good codes that
are not full diversity.

\IEEEtriggeratref{68}


\end{document}